\documentclass[runningheads]{llncs}
\usepackage[T1]{fontenc}
\usepackage{graphicx}
\usepackage{subcaption}
\usepackage{algorithmic}
\usepackage{amsmath}
\usepackage{amsfonts}
\usepackage{amssymb}
\usepackage{url}
\usepackage{bbding}
\usepackage{orcidlink}

\DeclareMathOperator{\logit}{logit}
\DeclareMathOperator{\expit}{expit}
\DeclareMathOperator{\var}{Var}

\newcommand\p[1]{{\left(#1\right)}}
\newcommand\bracket[1]{{\left[#1\right]}}
\newcommand\mynorm[1]{{\left\|#1\right\|}}

\begin{document}
\title{Revisiting the attacker's knowledge in inference attacks against Searchable Symmetric Encryption}
\titlerunning{Revisiting the attacker's knowledge in SSE attacks}
%
\author{Marc Damie\inst{1,2}\orcidlink{0000-0002-9484-4460}\thanks{Corresponding author: \email{m.f.d.damie@utwente.nl}}
    \and
    Jean-Benoist Leger\inst{3,4}\orcidlink{0000-0003-3409-5303} \and\\
    Florian Hahn \inst{2}\orcidlink{0000-0003-4049-5354} \and
    Andreas Peter \inst{5}\orcidlink{0000-0003-2929-5001}}
\authorrunning{M. Damie et al.}

\institute{Inria, France \and
    University of Twente, The Netherlands\and
    Université de technologie de Compiègne, CNRS, Heudiasyc, France \and
    Université Paris-Saclay, AgroParisTech, INRAE, UMR MIA Paris-Saclay \and
    Carl von Ossietzky Universität Oldenburg, Germany}
\maketitle              
\begin{abstract}
    Encrypted search schemes have been proposed to address growing privacy concerns.
    However, several leakage-abuse attacks have highlighted some security vulnerabilities.
    Recent attacks assumed an attacker's knowledge containing data ``similar'' to the indexed data.
    However, this vague assumption is barely discussed in literature: how likely is it for an attacker to obtain a ``similar enough'' data?

    Our paper provides novel statistical tools usable on any attack in this setting to analyze its sensitivity to data similarity.
    First, we introduce a mathematical model based on statistical estimators to analytically understand the attackers' knowledge and the notion of similarity.
    Second, we conceive statistical tools to model the influence of the similarity on the attack accuracy.
    We apply our tools on three existing attacks to answer questions such as: is similarity the only factor influencing accuracy of a given attack?
    Third, we show that the enforcement of a maximum index size can make the ``similar-data'' assumption harder to satisfy.
    In particular, we propose a statistical method to estimate an appropriate maximum size for a given attack and dataset.
    For the best known attack on the Enron dataset, a maximum index size of 200 guarantees (with high probability) the attack accuracy to be below 5\%.
    \keywords{Searchable Encryption  \and Attacks \and Statistics.}
\end{abstract}

\section{Introduction}
With the increasing popularity of cloud services, there is a growing concern about the privacy and confidentiality issues induced by such practices.
Song et al. \cite{song_practical_2000} proposed the first construction to search on encrypted data.
Curtmola et al. \cite{curtmola_searchable_2006} later used this result to build a Searchable Symmetric Encryption (SSE) scheme.
This scheme enables an efficient keyword search across encrypted documents.

Such SSE schemes build an encrypted index that can be queried to obtain the (encrypted) documents containing the queried keyword.
These schemes leak information exploitable by an attacker to recover the plaintext query.
Our work studies the single-keyword search SSE schemes leaking access and search patterns.
Our focus will be on the type of attack that received the most attention: passive attacks assuming attacker-known data \cite{cash_leakage-abuse_2015,damie_highly_2021,dijkslag_passive_2022,gui_rethinking_2023,islam_access_2012,lambregts_volume_2022,oya_hiding_2021,oya_ihop_2022,pouliot_shadow_2016,xu_leakage-abuse_2023}.
These attacks build a co-occurrence matrix from the leakage and compare it to a keyword co-occurrence matrix computed on an attacker-known dataset.

\paragraph{Passive query-recovery attacks}
Blackstone et al. \cite{blackstone_revisiting_2020} defined two types of passive query-recovery attacks against SSE: inference attacks and known-data attacks.
The known-data attacks assume that the attacker's data is indexed.
Inference attacks \cite{blackstone_revisiting_2020,damie_highly_2021} only require an attacker's knowledge ``similar but different'' from the indexed data (as formulated in \cite{damie_highly_2021}); i.e., cover a more general setting.
The first SSE attacks \cite{blackstone_revisiting_2020,cash_leakage-abuse_2015,islam_access_2012,lambregts_volume_2022,ning_passive_2018,ning_leap_2021leap,poddar_practical_2020} were only usable (or only accurate) as known-data attacks.
Recently, Damie et al. \cite{damie_highly_2021} and Oya and Kerschbaum \cite{oya_ihop_2022} respectively presented the \textit{refined score attack} and the \textit{IHOP attack}, which \textbf{both achieved high recovery rates as inference attacks}.

Other recent works improved existing attacks.
Gui et al. \cite{gui_rethinking_2023} introduced an attack built similarly to the historical attack of \cite{islam_access_2012}, and takes leakage mitigation into account to improve the attack results.
Ho et al. \cite{ho_similar_2024} optimized the refined score attack using an additional volume leakage.

Other types of attack exist, but they are out of scope: active attacks \cite{langhout_file-injection_2024,zhang_inject_2024,zhang_high_2023,zhang_all_2016}, attacks exploiting the query frequency \cite{liu_search_2014,oya_hiding_2021}, and attacks against other schemes \cite{falzon_full_2020,grubbs_pump_2018,kellaris_generic_2016,kornaropoulos_data_2019,kornaropoulos_state_2020,lacharite_improved_2018,markatou_attacks_2023,markatou_reconstructing_2021}.
Our work focuses on (passive) inference attacks against static SSE schemes with access and search pattern leakage.

\paragraph{The ``similar-data'' assumption}
Inference attacks require the attacker's data to be ``sufficiently similar'' to the indexed data (as formulated in \cite{damie_highly_2021}).
However, no existing work discussed whether this vague assumption was realistic.

Existing works \cite{damie_highly_2021,dittert_too_2023} highlighted the (evident) link between dataset similarity and attack accuracy.
While this observation is interesting, it does not answer a core question: \textbf{how likely is it for an attacker to obtain a ``sufficiently similar'' dataset?}
This question is the \emph{starting point of our paper}.

\paragraph{Related works}
Recent works in SSE aimed at assessing the practicality of the attacks.
On the one hand, Kamara et al. \cite{kamara_sok_2022} developed a Python framework to standardize the attack simulation commonly used in attack papers.
On the other hand, multiple works provided different perspective on the leakage of SSE schemes \cite{boldyreva_understanding_2024,kamara_bayesian_2023,kornaropoulos_leakage_2022}.
While these works analyzed the SSE leakage, we study another part of the attacker's knowledge: the attacker's dataset.

\paragraph{Our contributions}
Our work provides several statistical tools to better analyze the sensitivity of SSE attacks to data similarity.
In particular, we propose:
\begin{enumerate}
    \item A \emph{mathematical model for the attacker's knowledge} based on statistical estimators. This enables an analytical study of the similar-data assumption.
    \item \emph{A statistical method to analyze the accuracy of any attack based on a similarity metric}. We apply it to compare three existing attacks and study questions such as: is similarity the only factor influencing the attack accuracy.
    \item Based on our analysis of the similarity, we show that \emph{setting a maximum index size can limit the attack accuracy}. We also provide a statistical method to estimate a maximum size based on a specific attack and dataset.
    \item A method to simulate attacker's data issued from a data breach; a recurrent motivation for the existence of similar data. Applying it on the Apache dataset shows that data breaches do not always provide \emph{similar enough} data.
\end{enumerate}

\section{Preliminaries}
\label{sec:Preliminaries}

\subsection{Searchable symmetric encryption (SSE)}
From a high-level point of view, SSE schemes are based on the design described in \cite{curtmola_searchable_2006}.
The client owns a dataset $ D_\text{ind}$.
For each $d\in D_\text{ind}$, we denote $\text{id}(d)$ its identifier.
We assume that $\text{id}(\cdot)$ leaks no information about $d$.
The client generates an inverted index and encrypts it using her secret key $K$.
Then, this encrypted index is uploaded to the server.
A query token is an output of the query function (denoted as $\text{Query}_K(w)$) taking as input the keyword $w$ and the secret key $K$.
This index associates query tokens with the identifiers of the documents containing the underlying keyword.
To query the keyword $w$, the client computes $\text{Query}_K(w)$ and sends this token to the server.
The server sends back the matching encrypted documents.

This work focuses on efficient SSE constructions \cite{bost_sophos_2016,bost_forward_2017,cash_dynamic_2014,chase_structured_2010,curtmola_searchable_2006,ghareh_chamani_new_2018,sun_practical_2021,sun_practical_2018} that leak the search and access patterns. 
Some works \cite{blackstone_revisiting_2020,kamara_structured_2018} proposed the avoidance of access pattern leakage, but its instantiation is inefficient.
Hence, our paper only studies SSE schemes leaking search and access patterns.
This leakage profile has been the target of many passive attacks \cite{cash_leakage-abuse_2015,damie_highly_2021,dijkslag_passive_2022,gui_rethinking_2023,islam_access_2012,lambregts_volume_2022,oya_hiding_2021,oya_ihop_2022,pouliot_shadow_2016,xu_leakage-abuse_2023}.

\subsection{Threat model}
A passive attacker can observe and link each query to its response.
Thus, the attacker can create (Query, DocIDs) pairs for each query observed.
The access pattern leakage corresponds to the set of matching (encrypted) documents leaked by each query.
The search pattern leakage lets an attacker identify if a query has been issued twice.
Hence, each keyword $w$ has a unique query token $\text{Query}_K(w)$.

We study the case of an \textit{honest-but-curious} server storing the encrypted index and following the protocol while trying to recover the keywords being queried.
Such an attacker can record and analyze the protocol transcript.
Nearly all attack papers used this setting; only \cite{langhout_file-injection_2024,zhang_inject_2024,zhang_high_2023,zhang_all_2016} proposed active attacks.
Finally, the attacker can use leaked data as long as they is not actively involved in the leak (e.g., a leaked dataset published by hackers).

Almost all passive attacks targeted SSE schemes supporting single-keyword search in a static setting (i.e., no update operation).
On the one hand, only Dijkslag et al. \cite{dijkslag_passive_2022} studied attacks on (static) schemes supporting conjunctive keyword search.
On the other hand, only Xu et al. \cite{xu_leakage-abuse_2023} presented a passive attack in a dynamic setting.
Therefore, our paper focuses on static schemes with single-keyword search because we want to study the most common attack setting.

\subsection{Attacker's knowledge}
\label{subsec:adv_knowledge}

Let $\text{Coocc}( D, w_a, w_b)$ be the function returning the number of documents from $ D$ in which both $w_a$ and $w_b$ appear (i.e., the number of co-occurrences).

\paragraph{Keyword universe} The attacker knows $\mathcal{W} = \{w_1, \dots, w_m\}$, the set of queryable keywords.

\paragraph{Attacker's dataset} The attacker knows a dataset $ D_\text{atk} = \{d_1, \dots, d_{n_\text{atk}} \}$.
The documents are ``\textit{similar but different}'' from the indexed documents.
Similar documents can be, for example, documents that have been leaked (then removed from the index).
From $ D_\text{atk}$, the attacker obtains the $m \times m$ keyword co-occurrence matrix defined as follows: $C^\text{atk}_{ij} = \text{Coocc}( D_\text{atk}, w_i, w_j)$.

\paragraph{Observed queries} The attacker has observed $l$ unique queries $\mathcal{Q} = \{q_1, \dots, q_l \}$ and their results.
We denote as $R(q) = \{id(d) | d \in  D_\text{ind} \wedge q = \text{Query}_K (w) \wedge w \in d\}$ the set of document identifiers returned for the query $q$.
The attacker can compute the $l \times l$ query co-occurrence matrix $C^\text{query}_{ij} =  \left|R(q_i) \cap R(q_j)\right|$.

\paragraph{Known queries} In \cite{cash_leakage-abuse_2015,damie_highly_2021,islam_access_2012}, the attacker has some \textit{known queries}: for $k$ different queries, the attacker knows the underlying keyword of the query.

\paragraph{Number of indexed documents} The attacker knows the number of documents indexed $n_\text{ind}$.
An honest-but-curious server storing the index can infer this.

\subsection{Attacks}
The attack algorithm \emph{takes as input} the attacker's knowledge presented in Section \ref{subsec:adv_knowledge}.
The algorithm \emph{outputs} a predicted keyword for each observed query.
The attack accuracy is the proportion of correctly predicted keywords.

\paragraph{Refined score attack} This attack \cite{damie_highly_2021} uses the co-occurrence matrices and the known queries to extract a vector characterizing each keyword and each query.
Each keyword vector is characterized by its co-occurrence with the keywords from the known queries.
Each query vector is characterized by its co-occurrence with the queries from the known queries.
The score attack identifies, for each query, the keyword minimizing the keyword-query distance (based on their respective vectors).
The refined score attack improves this naive idea via an iterative process to improve the query recovery.

\paragraph{IHOP attack} \cite{oya_ihop_2022} also uses the co-occurrence matrices but does not require known queries contrary to \cite{damie_highly_2021}.
This paper formulates the attack problem as an assignment problem (i.e., assigning queries to keywords) and solve it using an efficient iterative algorithm based on linear assignment.
While the refined score attack uses a vector distance to quantify the cost of an assignment, IHOP relies on log-likelihood functions.

\paragraph{Empirical conclusions}
Section \ref{sec:atk_analysis_comp} and \ref{sec:dataset_size_acc} present robust statistical methods to obtain empirical conclusions (e.g., parameters influencing the attack accuracy).
We demonstrate our methods on state-of-the-art attacks (i.e., the Score, Refined Score, and IHOP attacks), but these methods are usable on any future passive query-recovery attacks.
Thus, our empirical conclusions hold for these attacks specifically, and our methods should be applied on any future attack to confirm the conclusions on them.

\emph{Empirical conclusions are inherent to statistical approaches}: since statistics draw conclusions based on experimental data, the conclusions hold for any data drawn from the same distribution... any new data distribution (i.e., new attack in our case) requires to reproduce the analysis.
Conclusions conditioned to a data distribution is then common in statistics, and they have a clear motivation: they enable \emph{studying phenomena too complex to be analyzed analytically}.

Our main goal is to provide novel statistical tools to analyze and understand attacks against SSE.
Our empirical conclusions demonstrate the potential of our statistical tools and put a new light on the existing literature.
Nevertheless, note that conclusions from Section \ref{sec:noisy-knowledge} hold for all attacks because they rely on theoretical analysis.

\subsection{Experimental setup}
We perform our experiments on Debian Bullseye with a 4-core processor and 8 GB of memory.
We use three datasets: the commonly used Enron email dataset \cite{klimt2004introducing} (30,109 emails contained from \textit{\_sent\_mail} folders), the Apache mailing dataset \footnote{\url{http://mail-archives.apache.org/mod_mbox/lucene-java-user/}} (the 50,878 emails from the ``java-user'' mailing list from the Lucene project for 2002-2011), and the Blog Authorship dataset \footnote{Dataset archive: \url{https://web.archive.org/web/20200121222642/http://u.cs.biu.ac.il/~koppel/blogs/blogs.zip}} \cite{schler2006effects} (681,288 blog posts written by 19,320 authors), never used in an existing attack paper. Due to hardware constraints, we only use 200K posts simultaneously (picked uniformly at random) from this last dataset.

Since we are studying inference attacks, the dataset is split into two disjoint subsets of varying sizes to create $ D_\text{atk}$ and $ D_\text{ind}$.
By default, we perform this split uniformly at random, but Section \ref{sec:dataset_generation} explores an alternative splitting method.

We only extract keywords from the document content.
The keyword list is then obtained after stemming the words using the Porter stemmer \cite{porter1980algorithm} and removed the stop words (e.g., ``the'' or ``and'').
In the Apache dataset, we remove the mailing list signature from each email.
To generate the keyword universe $\mathcal{W}$, we choose the $m$ most frequent keywords in the complete dataset.

Like in \cite{damie_highly_2021,oya_ihop_2022}, the attack accuracy is the proportion of correct keyword recovery among a set of \textit{unique} observed queries.

Our Python codebase is publicly available here: \url{https://github.com/MarcT0K/Statistical-Leakage-SSE-attacks}.
We refactored the source code of the attacks published by \cite{damie_highly_2021,oya_ihop_2022}\footnote{Source code: \url{https://github.com/MarcT0K/Refined-score-atk-SSE}, \url{https://github.com/simon-oya/ihop-code}}.

\section{Revisiting the notion of similar data}
\label{sec:noisy-knowledge}

\paragraph{Essential notations}
Let $D_\text{atk}$ (resp. $D_\text{ind}$) be the attacker's dataset (resp. indexed dataset) and $n_{\text{atk}}$ (resp. $n_{\text{ind}}$), its size.
The index can be queried on any word $w$ contained in the keyword universe $\mathcal{W}$ (which contains $m$ keywords).
The keyword co-occurrence matrix $C^\text{atk}$ (resp. $C^\text{ind}$) is extracted from $D_\text{atk}$ (resp. $D_\text{ind}$).
The matrices are built based on the keyword universe $\mathcal{W}$, so they are of dimension $m\times m$.
Let $y \xleftarrow{R} \mathcal{Y}$ denote the sampling of $y$ from the
probability distribution $\mathcal{Y}$, and $p^X_{ij}$ denotes the following probability $\mathbb{P}(\text{keywords } (i,j) \text{ appear in } d \xleftarrow{R}  D_{X})$.

\subsection{Dataset similarity}
\label{subsec:sim_def}
Damie et al. \cite{damie_highly_2021} introduced the notion of dataset similarity to quantify the divergence between the attacker's dataset and the indexed dataset.

Let $C^\text{atk}$ (resp. $C^\text{ind}$) be an $m \times m$ co-occurrence matrix such that: $C^\text{atk}_{ij} = \text{Coocc}( D_\text{atk}, w_i, w_j)$ (resp. $C^\text{ind}_{ij} = \text{Coocc}( D_\text{ind}, w_i, w_j)$) for all keywords $w_i,w_j \in \mathcal{W}$.
The query co-occurrence matrix $C^\text{query}$ is a restriction of $C^\text{ind}$ with an unknown permutation (NB: the attacks aim to find the permutation): if $q_a=\text{Query}_K(w_i)$ and $q_b=\text{Query}_K(w_j)$ then $C^\text{query}_{ab} = C^\text{ind}_{ij}$.
Damie et al. \cite{damie_highly_2021} defined the $m \times m$ similarity matrix of $ D_\text{ind}$ and $ D_\text{atk}$ over the keyword universe $\mathcal{W}$:

\begin{equation}
    \label{eq:sim_mat_def}
    \text{SimMat} = \frac{C^\text{ind}}{n_\text{ind}} - \frac{C^\text{atk}}{n_\text{atk}}
\end{equation}

Using this matrix, they proposed a similarity metric $\epsilon$:
\begin{definition}[$\epsilon$-similarity]
    \label{def:eps-sim}
    $ D_\text{ind}$ and $ D_\text{atk}$ are \textit{$\epsilon$-similar} if\hspace{1em} $||\text{SimMat}|| \leq \epsilon$
\end{definition}

Definition \ref{def:eps-sim} uses the Frobenius norm (i.e., the equivalent of Euclidean norm for matrices) as matrix norm.
In the rest of the paper, we refer to $\epsilon$ as the smallest value that satisfies this inequality (i.e., $\epsilon = \mynorm{\text{SimMat}}$).

\paragraph{Alternative metrics}
Gui et al. \cite{gui_rethinking_2023} introduced two other similarity metrics: ``absolute distance'' and ``Modified Probability Score''.
The absolute distance has the same formula as the $\epsilon$-similarity but uses the infinity norm instead of the Frobenius norm.
The Modified Probability Score relies on conditional probabilities.

\paragraph{Metric choice}
There is no universally better metric; the choice of metric is subjective.
Our paper uses the $\epsilon$-similarity because its formula is convenient for mathematical analysis.
In other contexts, the Modified Probability Score might be preferable: e.g., if the link between the metric and conditional probability is necessary for mathematical analysis.

\paragraph{Dataset similarity assumption}
All inference attacks assumes that there exists a sufficiently small $\epsilon$ such that $D_\text{ind}$ and $D_\text{atk}$ are $\epsilon$-similar.
If we assume no similarity, the attacker could use random data.
Our work will explore what is a ``sufficiently small'' $\epsilon$, and whether such similarity can be obtained in practice.

\subsection{The attacker's knowledge is inherently noisy}
\label{subsec:coocc}
To analyze the similarity assumption, we first need a proper mathematical model for the co-occurrence matrices.
As in machine learning (ML), we consider a dataset as a sample of a random distribution.
The properties of the underlying probability distribution could bring more or less uncertainty depending on the use case.
We will model this uncertainty and show its impact on the similarity assumption.

\paragraph{Mathematical model}
Let $ D_{\text{atk}}$ and $ D_{\text{ind}}$  be two datasets composed of documents that are represented as binary vectors of length $m$.
Each vector component $i$ indicates whether the keyword $w_i$ is contained in the document.
Furthermore, let $\mathcal{X}_\text{atk}$ denote the random variable that describes the experiment of sampling a document (i.e., a  binary vector of length $m$) from the same probability distribution as given by the documents in $D_{\text{atk}}$.
In other words, $\mathcal{X}_\text{atk}$ is a vector of dependent Bernoulli random variables.
It implies that the co-occurrence matrix $C^\text{atk}$ is composed of realizations of dependent Binomial variables, where $C^\text{atk}_{ij} \xleftarrow{R} \mathcal{C}^{n_\text{atk}, p_\text{atk}}_{ij} = \mathcal{B}(n_\text{atk}, p_{ij}^\text{atk})$.
The probability $p^\text{atk}_{ij}$ corresponds to the probability that the keywords $i$ and $j$ both appear in a document from $ D_\text{atk}$.

Acknowledging the dependence between the Binomial variables is essential because it significantly complicates the mathematical analysis.
To be convinced of the dependence, let us consider the probabilities $p_{ij}$ and $p_{ii}$ (i.e., the probability of keyword $i$ appearing in a document).
We have $p_{ij} \le p_{ii}$ because the event ``keyword $i$ appears in a document'' contains the event ``keywords $i$ and $j$ appears in the same document''.

Let $\mathcal{X}_\text{ind}$ be the random variable for dataset $ D_{\text{ind}}$ defined with the same procedure as $\mathcal{X}_\text{atk}$ for dataset $ D_{\text{atk}}$.
Similarly, we deduce that the co-occurrence matrix $C^\text{ind}$ is composed of realizations of dependent Binomial variables, where $C^\text{ind}_{ij} \xleftarrow{R} \mathcal{C}^{n_\text{ind}, p_\text{ind}}_{ij}=\mathcal{B}(n_\text{ind}, p_{ij}^\text{ind})$.

To sum up, $\mathcal{C}^{n_\text{ind}, p_\text{ind}}$ and $\mathcal{C}^{n_\text{atk}, p_\text{atk}}$ are two random matrix probability distributions from which $C^\text{ind}$ and $C^\text{atk}$ are drawn.
We assume the two distributions to be independent because, in inference attacks, the attacker's dataset is considered ``different'': implicitly, it means obtained from independent sources.

\paragraph{Statistical estimators}
The previously defined random distribution ($\mathcal{C}^{n_\text{ind}, p_\text{ind}}$ and $\mathcal{C}^{n_\text{atk}, p_\text{atk}}$) \textbf{cannot be known exactly}.
We are in a classic statistics problem: we do not have access to the exact probabilities, but we know a dataset drawn from an \textit{unknown} probability distribution.
We can use the dataset to compute an \textit{estimation} of an \textit{unknown} probability.
The interest of statistical estimation is then to approximate these unknown probabilities.

Finding a similar dataset is close to the problem of representative sampling for surveying: the results are not the probabilities themselves (e.g., political opinions) but an estimation of this probability.
The larger the sample is, the more precise the estimation.
If one compares two samples obtained from populations with different probabilities, their respective estimators may be close to each other, but it is unlikely.
For example, a survey on professors might provide similar results to the same survey on students.
However, one understands that it is unlikely, especially with numerous questions (i.e., high-dimensional data).

Analogously, our attacker only observes samples from a random distribution (i.e., documents) but not the probabilities themselves.
They can estimate the co-probabilities\footnote{We use the term co-probability to refer to the probability of a keyword co-occurrence.} $p_{ij}^\text{atk}$ (resp. $p_{ij}^\text{ind}$) by computing the co-frequencies $\frac{C^\text{atk}_{ij}}{n_\text{atk}} = \widehat{p}_{ij}^\text{atk}$ (resp. $\frac{C^\text{ind}_{ij}}{n_\text{ind}}= \widehat{p}_{ij}^\text{ind}$)\footnote{The notation $\widehat{x}$ refers to the estimation of the unknown value $x$.}.
The co-frequency is the \textit{maximum likelihood estimator} of $p_{ij}$.
The maximum likelihood estimator is an efficient estimator\footnote{In statistics, the efficiency \cite{everitt_cambridge_2010} measures the quality of an estimator. Here, the maximum likelihood estimator reaches the Cramér-Rao lower bound \cite{cramer_mathematical_1999,rao_information_1945} which bounds the variance of an unbiased estimator.} for $p_{ij}$ and converges toward this (unknown) value.
The larger $ D_\text{atk}$ (resp. $ D_\text{ind}$) is, the more precise the estimation will be.

Using the statistical estimators, \textbf{we can reformulate the ``similar-data'' assumption}: an attack is successful if the statistical estimation of $p_{ij}^\text{ind}$ and $p_{ij}^\text{atk}$ are close enough (for all $i,j\in\{1\dots m\}$).
While this sounds nearly the same as the initial formulation, it opens new perspectives: the SSE attack problem boils down to a statistical estimation problem.
We can then reuse results and methods from statistics literature to understand the shortcomings of SSE attacks.

\subsection{The impact of dataset sizes on the similarity}
\label{subsec:size_similarity}
A key property of statistical estimators is their convergence when the sample size tends to infinity.
The sample size has then a direct impact on the estimation quality, and on the underlying similarity assumption.
In our case, the sample size corresponds to the number of documents in the dataset.
This subsection investigates \emph{analytically and experimentally} how the dataset sizes influence the dataset similarity (using the $\epsilon$-similarity).

\textbf{Equality of co-probabilities}
First, we make the following assumption: $\forall i,j \in [m], p_{ij}^\text{ind} = p_{ij}^\text{atk}$.
In other words, for all pairs of keywords $w_i$ and $w_j$ from $\mathcal{W}$, the probability of their joint appearance in $ D_\text{atk}$ is equal to the probability of the same event in $ D_\text{ind}$\footnote{The equality of co-probabilities does not imply that the distributions are equal (i.e., $\mathcal{X}^\text{ind} \sim \mathcal{X}^\text{atk}$) because it only fixes $\frac{m(m+1)}{2}$ parameters out of their $2^{m-1}$ parameters.}.

As shown in Appendix \ref{app:uniform_split}, this assumption is the most advantageous setup for an attacker because it induces a smaller $\epsilon$-similarity.
Our conclusions then hold for the most advantaged attacker, \emph{so they hold for any attacker.}

\textbf{Analysis of the $\epsilon$-similarity}
Let us start from the co-probability estimators (i.e., $\widehat{p}_{ij}^\text{ind} = \frac{C_{ij}^\text{ind}}{n_\text{ind}}$ and $\widehat{p}_{ij}^\text{atk} = \frac{C_{ij}^\text{atk}}{n_\text{atk}}$) which are central in the attacks.
Since they are maximum likelihood estimators, we have asymptotic normality of estimators (\cite{griliches1983handbook} Chapter 36, Theorem 3.3).
Considering $n_{atk}$ and $n_{ind}$ are sufficiently large, we can approximate the distribution at finite distance by the normal distribution.
Let $Z^\text{atk}_{ij} \sim \mathcal{N}(0, \sigma^2_{\text{atk},ij})$ and $Z^\text{ind}_{ij} \sim \mathcal{N}(0, \sigma^2_{\text{ind},ij})$ be two Gaussian distributions.
Under the equality of co-probabilities, we have:

\begin{align}
    \widehat{p}^\text{ind}_{ij} & = p_{ij} + \frac{1}{\sqrt{n_\text{ind}}}Z^\text{ind}_{ij}, \widehat{p}^\text{atk}_{ij} = p_{ij} + \frac{1}{\sqrt{n_\text{atk}}}Z^\text{atk}_  {ij}
\end{align}

It implies that $(\widehat{p}^\text{ind}_{ij} - \widehat{p}^\text{atk}_{ij}) = \sqrt{\frac{1}{n_\text{ind}} +  \frac{1}{n_\text{atk}}}Z_{ij},  Z_{ij} \sim \mathcal{N}(0, \sigma_{ij}^2)$.
We can reuse this result in the expression of $\epsilon$:
\begin{equation}\label{eq:doc_size_inf}
    \begin{split}
        \epsilon          & = \mynorm{\text{SimMat}} = \sqrt{\sum_{i,j} (\frac{C^\text{ind}_{ij}}{n_\text{ind}} - \frac{C^\text{atk}_{ij}}{n_\text{atk}})^2}   = \sqrt{\sum_{i,j} (\widehat{p}^\text{ind}_{ij} - \widehat{p}^\text{atk}_{ij})^2} \\
        \Rightarrow \epsilon & = \sqrt{\sum_{i,j}(\frac{1}{n_\text{ind}} +  \frac{1}{n_\text{atk}})(Z_{ij})^2} = \sqrt{(\frac{1}{n_\text{ind}} +  \frac{1}{n_\text{atk}})K}, K\text{ a random variable}
    \end{split}
\end{equation}

We can draw two conclusions from Equation \eqref{eq:doc_size_inf}.
First, the size of the attacker's dataset matters as much as the size of the indexed dataset.
It seems evident that the more documents the attacker knows, the better the attack is.
However, thinking that a bigger indexed dataset helps the attacker could seem counterintuitive.
To understand this intuition, we must remember that the attacks rely on comparing two estimators (i.e., the co-occurrence matrices).
The attack is unsuccessful if \emph{any of them} is poorly estimated.

Second, fixing a dataset size creates a threshold for similarity.
For example, let us fix $n_\text{ind}$ and observe the effect on Equation \eqref{eq:doc_size_inf}: we have $\lim_{n_\text{atk} \rightarrow\infty} \epsilon =  \sqrt{\frac{1}{n_\text{ind}}K}$.
Hence, even with an infinite-sized attacker's dataset, having an $\epsilon$-similarity below a certain threshold (proportional to $\sqrt{\frac{1}{n_\text{ind}}}$) is unlikely.

\begin{figure}
    \centering
    \begin{subfigure}{0.45\textwidth}
        \centering
        \includegraphics[width=\linewidth]{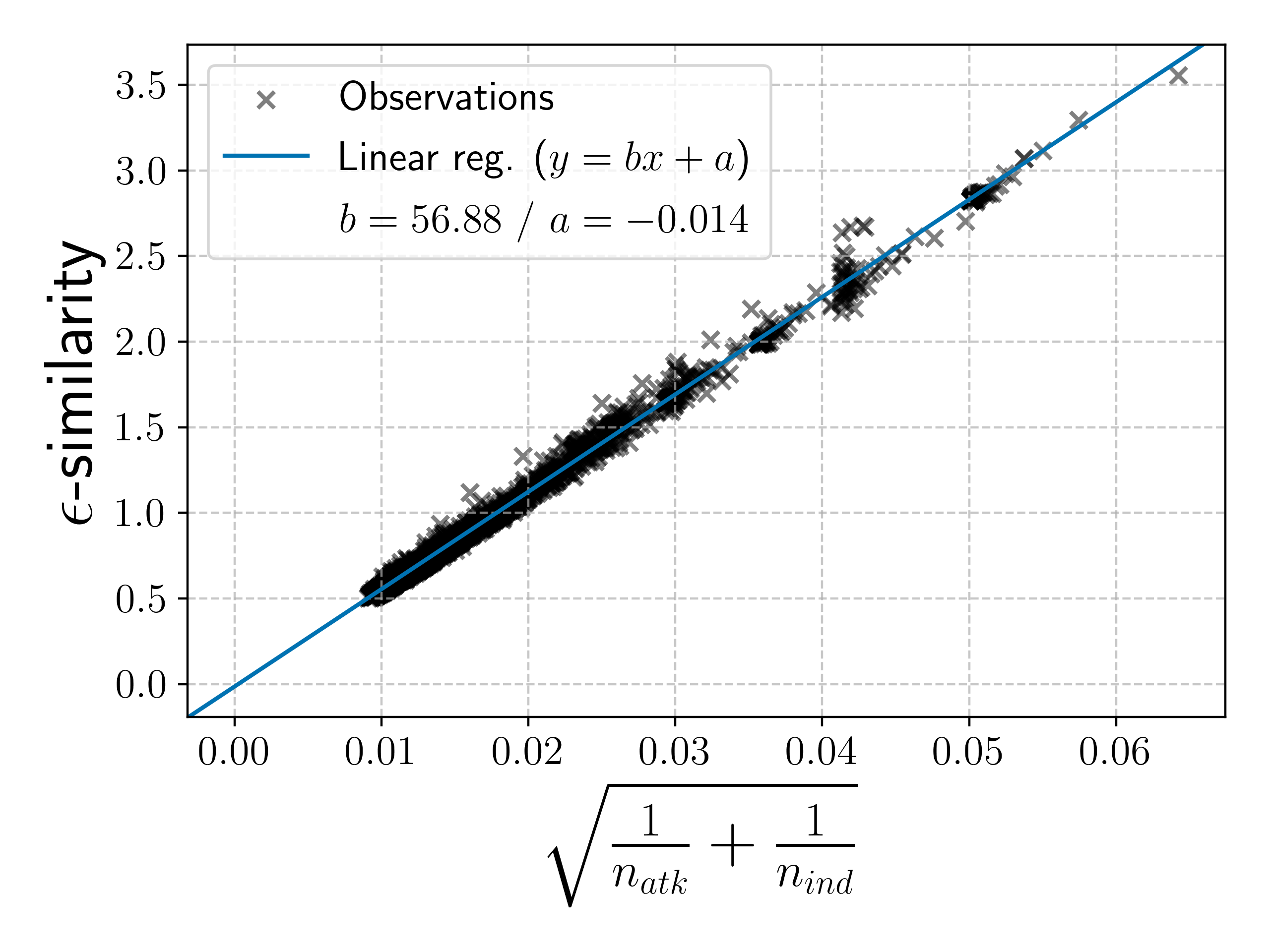}
        \caption{Varying $n_\text{atk}$ and $n_\text{ind}$}
        \label{fig:size_influence}
    \end{subfigure}
    \begin{subfigure}{0.45\textwidth}
        \centering
        \includegraphics[width=\linewidth]{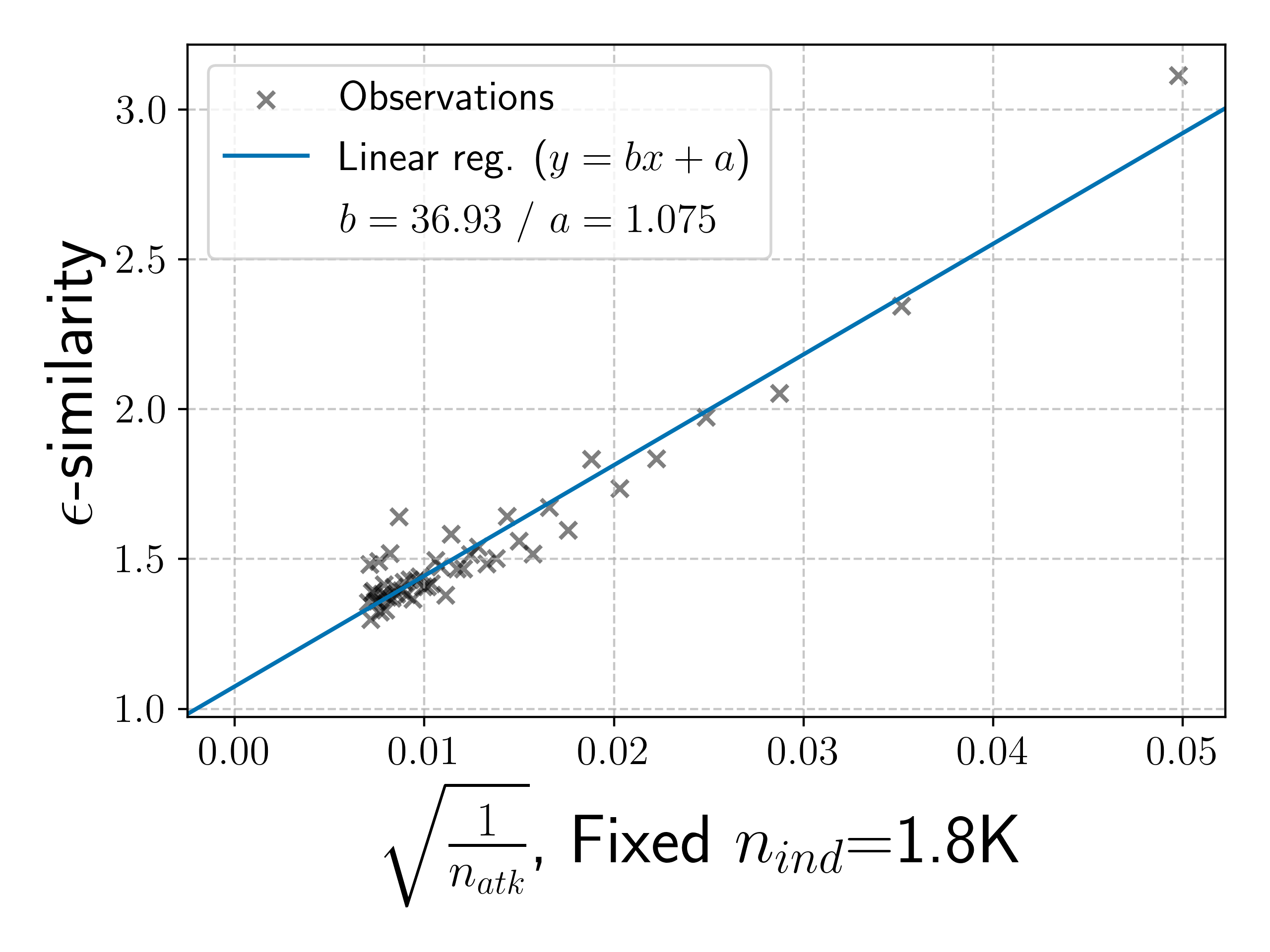}
        \caption{Varying $n_\text{atk}$, Fixed $n_\text{ind}$}
        \label{fig:size_influence_fixed}
    \end{subfigure}
    \caption{Influence of the dataset size on the similarity (Dataset: Apache)}
\end{figure}

\paragraph{Experimental confirmation}
Figure \ref{fig:size_influence} illustrates this theoretical result.
To generate results over a wide range of similarities, we vary the attacker's and indexed dataset sizes vary between 5\% and 95\% of the Apache dataset.
This plot comprises 2500 points with varying $n_\text{ind}$ and $n_\text{atk}$.
On the one hand, the linear relationship between $\sqrt{\frac{1}{n_\text{ind}} +  \frac{1}{n_\text{atk}}}$ and $\epsilon$ confirms the symmetric influence of these variables.
On the other hand, Figure \ref{fig:size_influence_fixed} shows the similarity threshold created when $n_\text{ind}$ is fixed.
Indeed, the linear relationship in Figure \ref{fig:size_influence} has an intercept tending to zero, while the linear relationship in Figure \ref{fig:size_influence_fixed} has an intercept equal to 1.
This non-zero intercept highlights the existence of a similarity threshold since, even with an infinite-sized attacker's dataset, the average $\epsilon$-similarity would be 1.
This threshold is \emph{specific to the  Apache dataset}; each has a different threshold influenced by the data distribution.

\section{Attack analysis from a data-similarity point of view}
\label{sec:atk_analysis_comp}

Existing works \cite{damie_highly_2021,dittert_too_2023} highlighted an evident link between similarity and attack accuracy: the more similar the (indexed and attacker) datasets are, the more accurate the attacks are.
This section goes beyond this simple intuition to fully explore the role of similarity in attack success.
We propose a robust statistical method to analyze the accuracy (of a given attack) using the similarity metric.
We then apply this method on three existing and three different datasets.

On the one hand, we use this method to understand whether similarity is the only parameter influencing the accuracy of an attack.
On the other hand, we use the method to build a more consistent attack comparison methodology.

\subsection{How to analyze attack success using a similarity metric?}
\label{subsec:sim_acc}
We want to estimate a function $\widehat{f}_\text{Acc}$ modelling the average attack accuracy in function of the $\epsilon$-similarity; i.e., $\widehat{f}_\text{Acc}(\epsilon) = \mathbb{E}(\text{Acc})$\footnote{$\mathbb{E}\left[X\right]$ is the expected value of the random variable $X$.}.
A continuous function provides a more detailed understanding of the attack's strengths and weaknesses.
Such functions enable extrapolating the experimental results and precisely identify gaps in the literature.
Attack papers usually represent the accuracy on a finite set of points, which gives a poor understanding of the complete attack behavior.
Sections \ref{subsec:sim-role} and \ref{subsec:comp} will show how we can use these functions to provide novel insights about the existing attacks.

\begin{figure*}[t]
    \centering
    \begin{subfigure}{0.45\textwidth}
        \centering
        \includegraphics[width=\linewidth]{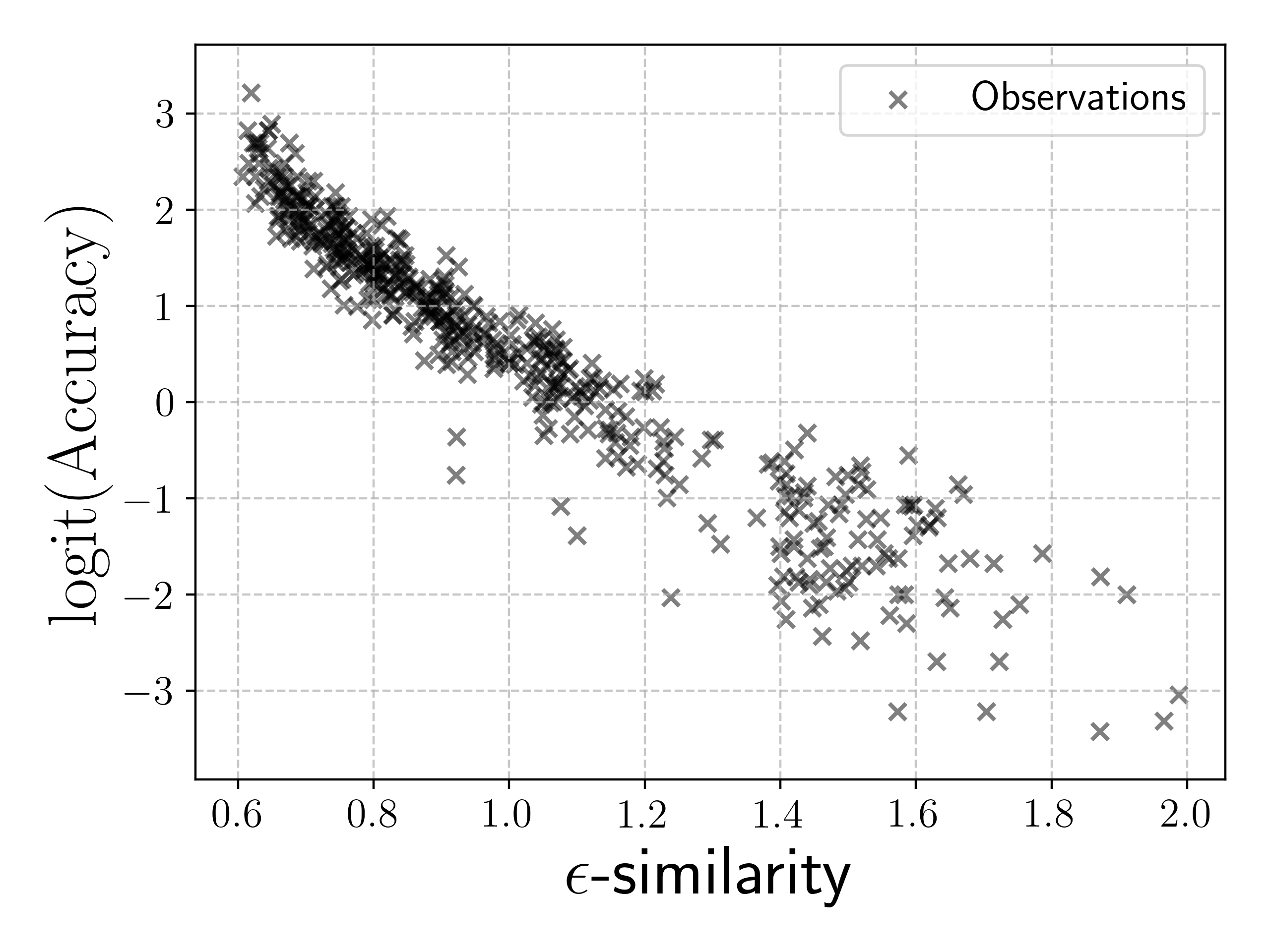}
        \caption{With $\logit$ transformation}
        \label{fig:atk_analysis_logit_enron}
    \end{subfigure}
    \begin{subfigure}{0.45\textwidth}
        \centering
        \includegraphics[width=\linewidth]{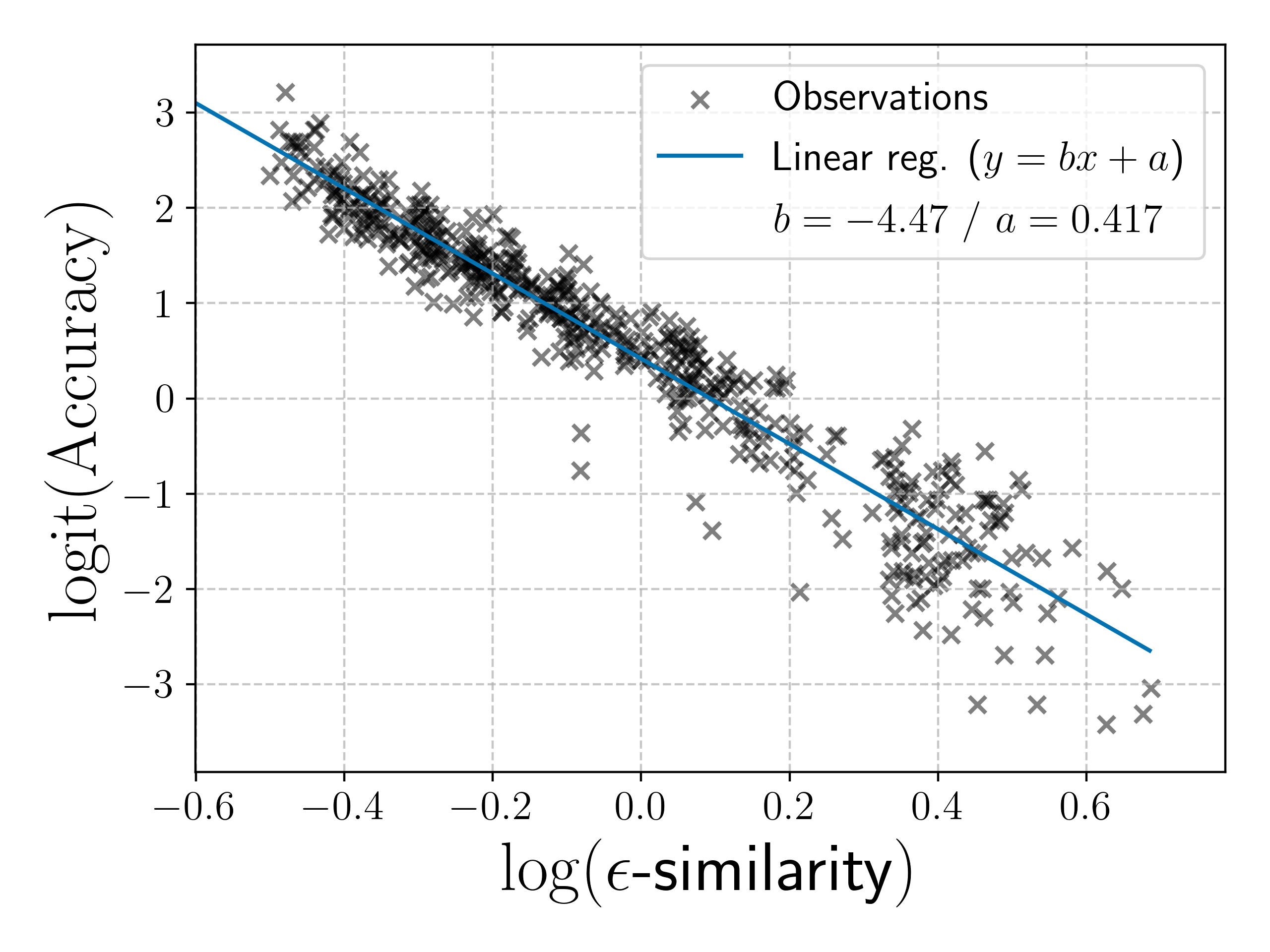}
        \caption{With $\logit-\log$ transformation}
        \label{fig:atk_analysis_logitlog_enron}
    \end{subfigure}
    \begin{subfigure}{0.45\textwidth}
        \centering
        \includegraphics[width=\linewidth]{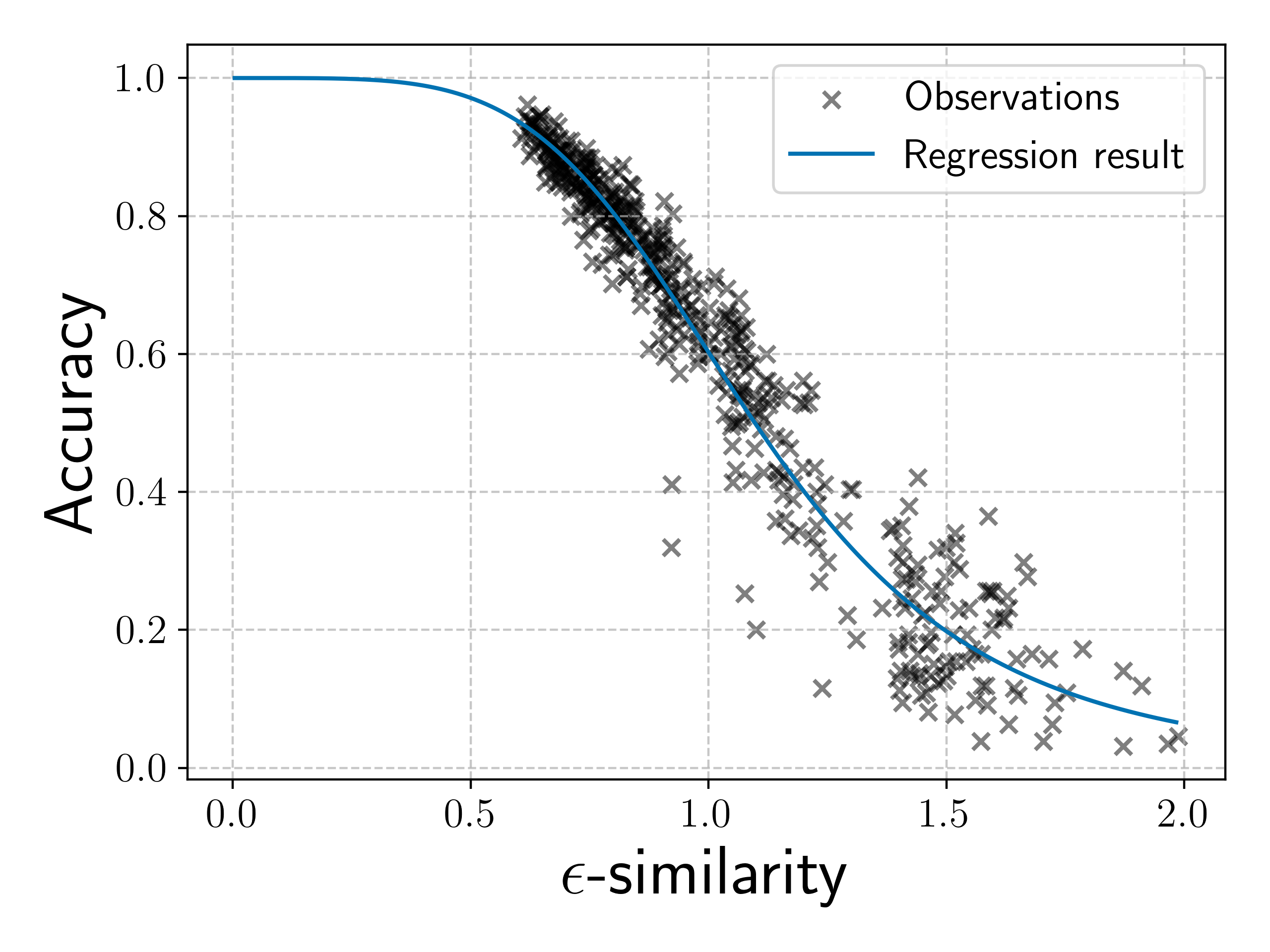}
        \caption{In the initial space}
        \label{fig:enron_atk_acc}
    \end{subfigure}
    \caption{Analysis of the refined score attack on Enron dataset ($m=1K$)}
    \label{fig:analysis_enron}
\end{figure*}

\paragraph{Estimating the accuracy function}
Before estimating the attack accuracy function, we need attack results.
Following the results from Section \ref{subsec:size_similarity}, we vary the size of the datasets to obtain attacks results with varying $\epsilon$-similarity.

Estimating a function based on this experimental data corresponds to a regression problem.
First, we can point out that the function cannot be linear since the accuracy is in $[0,1]$.
Hence, we cannot directly use linear regression on the raw data.
To circumvent this difficulty, we estimate a function $\widehat{f}_\text{Acc}'(\epsilon) = \logit(\text{Acc})$.
The $\logit$ function is defined as $\logit(p) = \log(\frac{p}{1-p})$ with $p\in\left(0,1\right)$ and has a known inverse function $\expit$.
This $\logit$ function maps a $\left(0,1\right)$ space into the real number space.
This logarithmic transformation is common in statistics.

Figure \ref{fig:atk_analysis_logit_enron} represents the simulation results with the $\logit$ transformation.
We observe two problems in this figure: we do not have apparent linearity between the variables, and there is a ``heteroscedasticity of the noises''.
The first problem is simple: the distribution of the points seems flattened when $\epsilon$ grows.
The second issue concerns the assumption made when computing a linear regression.
A fundamental assumption in linear regression is that the noise distribution is the same at each point of the space.
Then, the data should have the same noise distribution for small and high $\epsilon$, which is not the case in Figure \ref{fig:atk_analysis_logit_enron}.
A recurrent solution in ML to the linearity problem is to apply a logarithm transformation on the $x$ axis.
Figure \ref{fig:atk_analysis_logitlog_enron} presents the results of a linear regression between $\logit(\text{Accuracy})$ and $\log(\epsilon)$.
Combining these two logarithmic transformations makes the linear relationship more apparent.
Moreover, the logarithmic transformation slightly corrected the heteroscedasticity.
Scaling methods (where a data point is modified by a scaling factor depending on its $\epsilon$) exist to perfectly fix heteroscedasticity.
However, they add more complexity for equivalent results.

To sum up, we map the attack simulation results in a $\logit-\log$ space to compute a linear regression (i.e., $\logit(\mathbb{E}(\text{Acc})) = b\log(\epsilon) + a$).
We then deduce the average accuracy function $\widehat{f}_\text{Acc}(\epsilon) = \expit(b\cdot\log(\epsilon) + a)$ with $(b,a)$ the regression parameters.

Figure \ref{fig:analysis_enron} represents the resulting average accuracy function for the refined score attack (obtained on the Enron dataset).
This estimation protocol can be executed on any passive attack; e.g., Figure \ref{fig:comp_atk} compares the average accuracy function obtained on three different attacks.

\paragraph{Alternative regression methods}
Our regression model is built on classical statistical approaches (e.g., the $\logit$ function) and produces convincing results, but we cannot formally prove that it is the best model to represent the relationship between accuracy and similarity.
This situation is perfectly summarized by an aphorism attributed to the statistician George Box \cite{box_science_1976}: ``all models are wrong, but some are useful.''
In other words, it is commonly accepted that a statistical model cannot capture all the complexities of the real world.
However, this relative imperfection does not prevent from using statistical models.

The ML literature has described many other regression models, but our method is simple and provides interpretable results despite an initial non-linear problem.
A model is ``interpretable'' if an observer can understand the cause of a decision \cite{miller_explanation_2019} (e.g., predict the impact of an input variation).
This interpretability is absent in many popular ML models, such as neural networks and random forests.
These models can solve non-linear problems efficiently, but are often referred to as ``black boxes'' due to their lack of interpretability.
Our model is interpretable and efficient because it combines linear models with simple logarithmic transformations.
Interpretability is a crucial property in our case to foresee the security limits... making our model particularly useful.


\subsection{The role of similarity in attack success}
\label{subsec:sim-role}

We can now use our average accuracy functions to understand whether similarity is the only factor influencing the attack success.
We can reformulate it as follows: for each attack, is there a unique accuracy function $\widehat{f}_\text{Acc}(\epsilon)$ valid for all datasets?

\begin{figure}
    \centering
    \begin{subfigure}{0.45\textwidth}
        \centering
        \includegraphics[width=\linewidth]{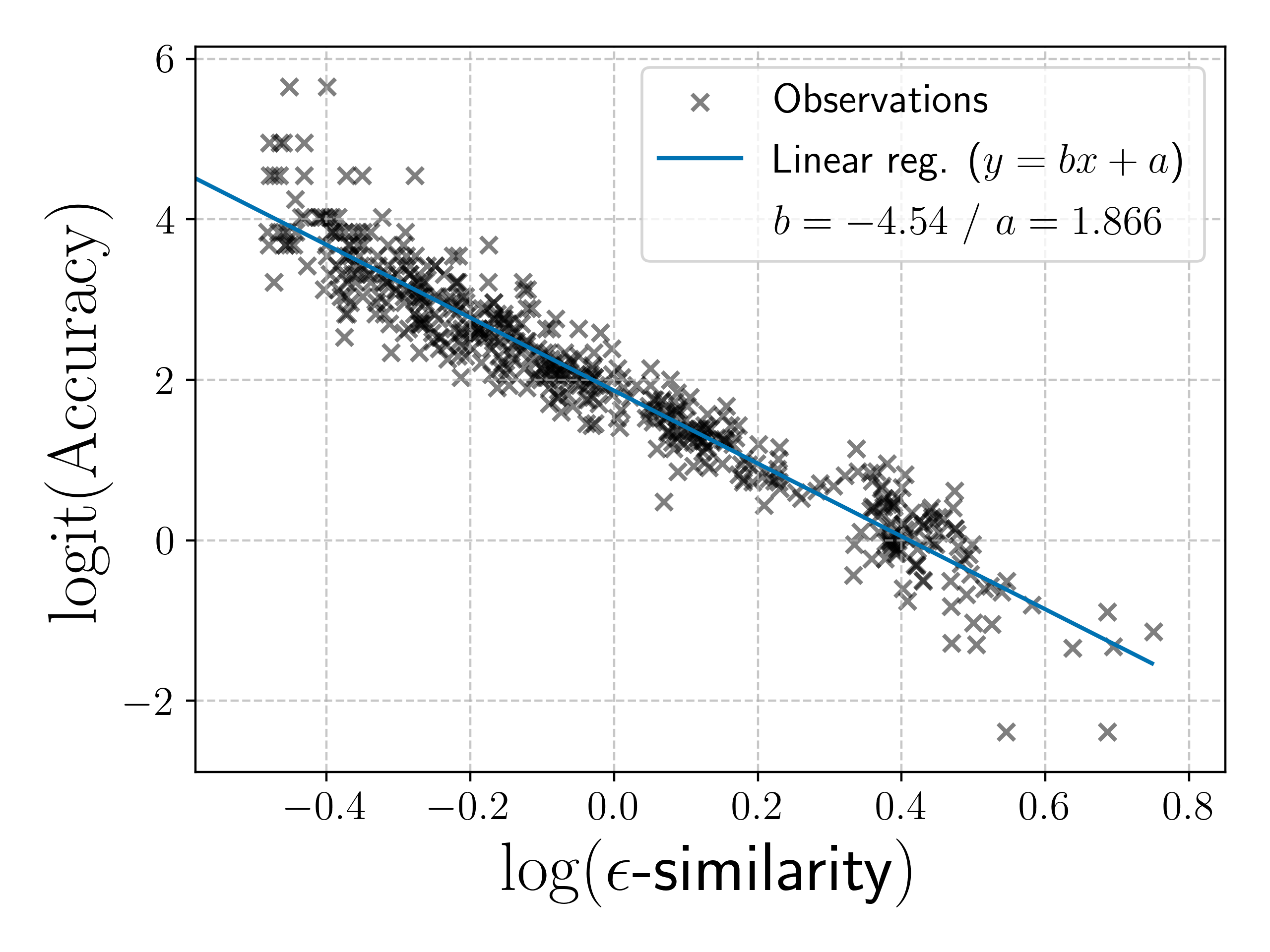}
        \caption{Apache dataset}
    \end{subfigure}
    \begin{subfigure}{0.45\textwidth}
        \centering
        \includegraphics[width=\linewidth]{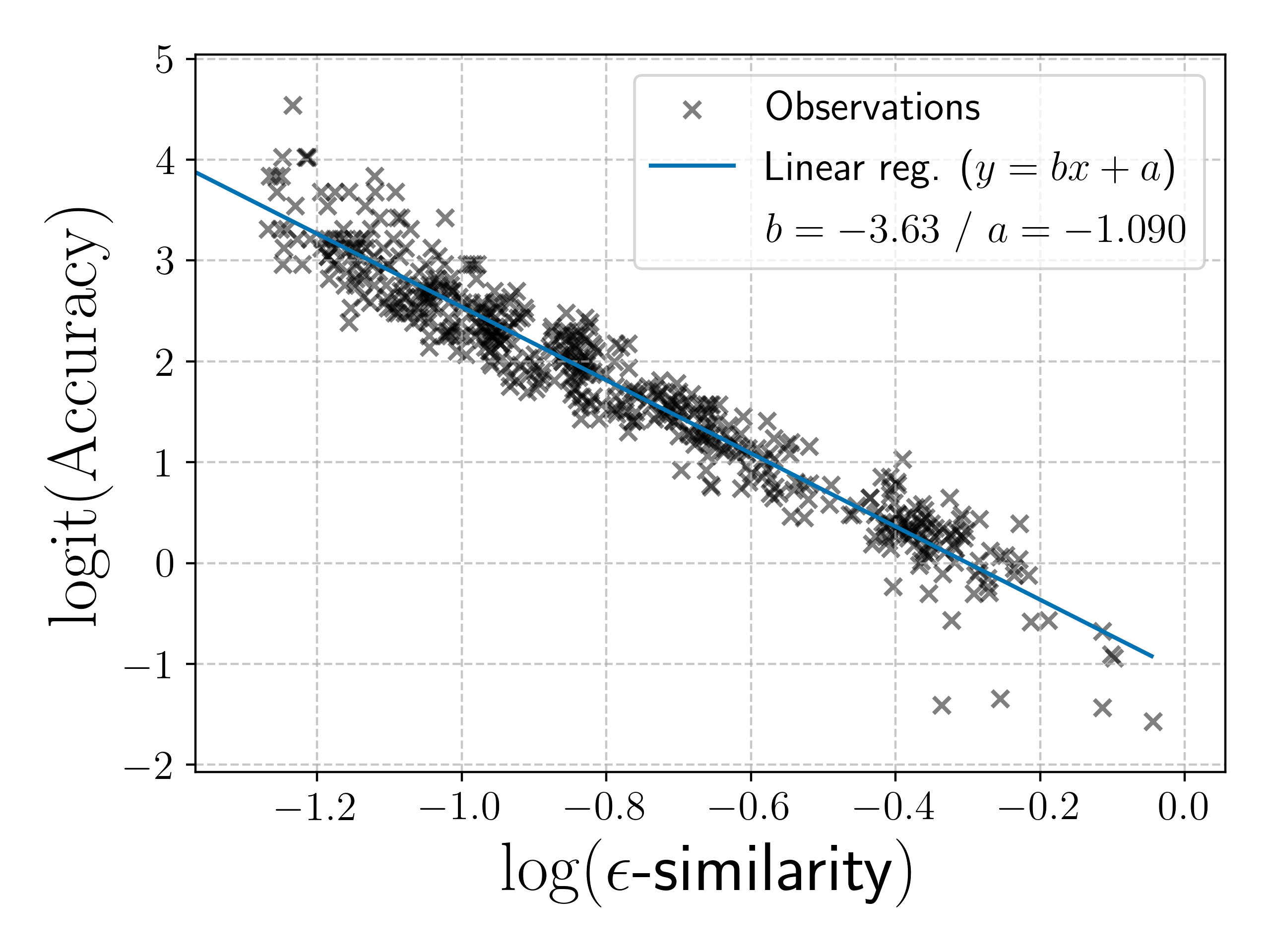}
        \caption{Blogs dataset}
    \end{subfigure}
    \caption{Regression results on others datasets ($m=1K$)}
    \label{fig:atk_analysis_apache_blogs}
\end{figure}

Figure \ref{fig:atk_analysis_apache_blogs} presents the average accuracy function on two other datasets: Apache and Blogs.
The linear regression parameters are $(-4.54, 1.87)$ for the Blogs dataset, $(-3.63, -1.09)$ for Apache, and $(-4.47, 0.417)$ for Enron.
These parameters significantly differ, so no unique function links the similarity to the accuracy.
Thus, \textbf{similarity is not the only parameter influencing the accuracy} of the refined score attack.

We can explain this phenomenon using keyword distribution.
Consider a hypothetical dataset with $m$ keywords: all keywords appear in all documents.
Even with perfect knowledge of the probabilities (i.e., $\epsilon=0$), we cannot distinguish the keywords with an accuracy better than a random guess (i.e., $\mathbb{E}(\text{Acc})=\frac{1}{m}$), because all keywords have the same co-frequencies.
Hence, each dataset has a keyword distribution, resulting in more or less distinguishable keywords.

\paragraph{Revisiting the notion of indistinguishability}
Indistinguishability is not a new notion in SSE \cite{bost_thwarting_2017,curtmola_searchable_2006,xu_interpreting_2021}.
Attack mitigations usually aim at some keyword/query indistinguishability.
The intuition is to have all queries leaking the same information, so the leakage induced by SSE schemes becomes useless to an attacker.

For example, mitigations based on false-positive results \cite{bost_thwarting_2017,cash_leakage-abuse_2015,xu_hardening_2019} (used in SSE implementations such as ShieldDB \cite{vo_shielddb_2023}) produces so-called ``indistinguishable'' queries by ``smoothing'' the keyword frequencies.
We argue that these countermeasures do not produce indistinguishability but simply noisy statistics (i.e., noisy co-occurrence estimators).
Indeed, the frequencies might be indistinguishable, but the co-frequencies do not automatically inherit this property.
To produce an indistinguishable leakage, we would need to smooth all the statistics, from the keyword frequencies to the co-frequencies of $m$ keywords (i.e., the $2^{m}-1$ parameters of the random binary vectors defined in Section \ref{subsec:coocc}).

However, the good performances of the mitigations show that noisy statistics are enough to prevent attacks.
We do not need indistinguishable leakages; we must only ensure they are realistically unusable.
For example, two queries can leak distinct co-frequency information; the attacks will be unsuccessful if this information leakage is too noisy.
Hence, it is unnecessary to perfectly protect the queries by enforcing them to have the same co-frequency leakage.
To reach a satisfying noise level, Section \ref{sec:dataset_size_acc} introduces a maximum index size estimation (guaranteeing a low-attack accuracy).


\subsection{Comparison framework}
\label{subsec:comp}

Finally, we can use the average accuracy functions to design a new attack comparison methodology: \emph{(1)} generate attack results for each attack, \emph{(2)} estimate the average accuracy functions, \emph{(3)} compare the average accuracy functions.

Figure \ref{fig:comp_atk} compares the accuracy functions of three recent attacks (on Apache dataset): Score, Refined score, and IHOP attacks.
The IHOP and refined score attacks outperform the score attack.
IHOP obtains slightly better results for smaller $\epsilon$ while for higher $\epsilon$, the refined score attack has a small advantage.

This comparison framework simplifies the identification of game-changing attacks.
The existing attack papers typically compare attacks over a set of parameters, emphasizing a clear accuracy difference between a novel attack and the state-of-the-art.
While this approach identifies improvements, it does not give a full picture of the situation.
In our case, a traditional attack comparison between IHOP and Refined Score attack would ``zoom in'' between $\epsilon=0.5$ and $\epsilon=1.0$ to focus on the highest accuracy gains.
This ``zoom'' ignores subtle phenomena, such as the slight advantage of the Refined Score attack on highly dissimilar datasets.
Systematically comparing accuracy functions would guarantee attack papers thoroughly analyze the attack behaviors and avoid focusing on a convenient set of parameters.
Moreover, the linear regression estimates an average distribution, smoothing the noise over the curve.
Noisy experimental results are a recurrent concern that questions whether an accuracy difference is significant.

Nevertheless, linear regression remains a statistical estimation, so some uncertainty remains.
We can model this uncertainty using confidence intervals.
As for individual experimental results, we can run simple statistical tests to prove a difference to be statistically significant.
While performing statistical tests on a large set of individual experimental results is tedious, performing a statistical test on linear regression parameters \cite{clogg_statistical_1995} is much easier.
In other words, we can statistically prove that an estimated function is significantly higher than another function instead of proving a significant difference on individual points.
Hence, this simplicity should \emph{encourage future attack papers to rigorously prove accuracy improvements using statistical tests on accuracy functions}.

\begin{figure}
    \centering
    \includegraphics[width=0.45\linewidth]{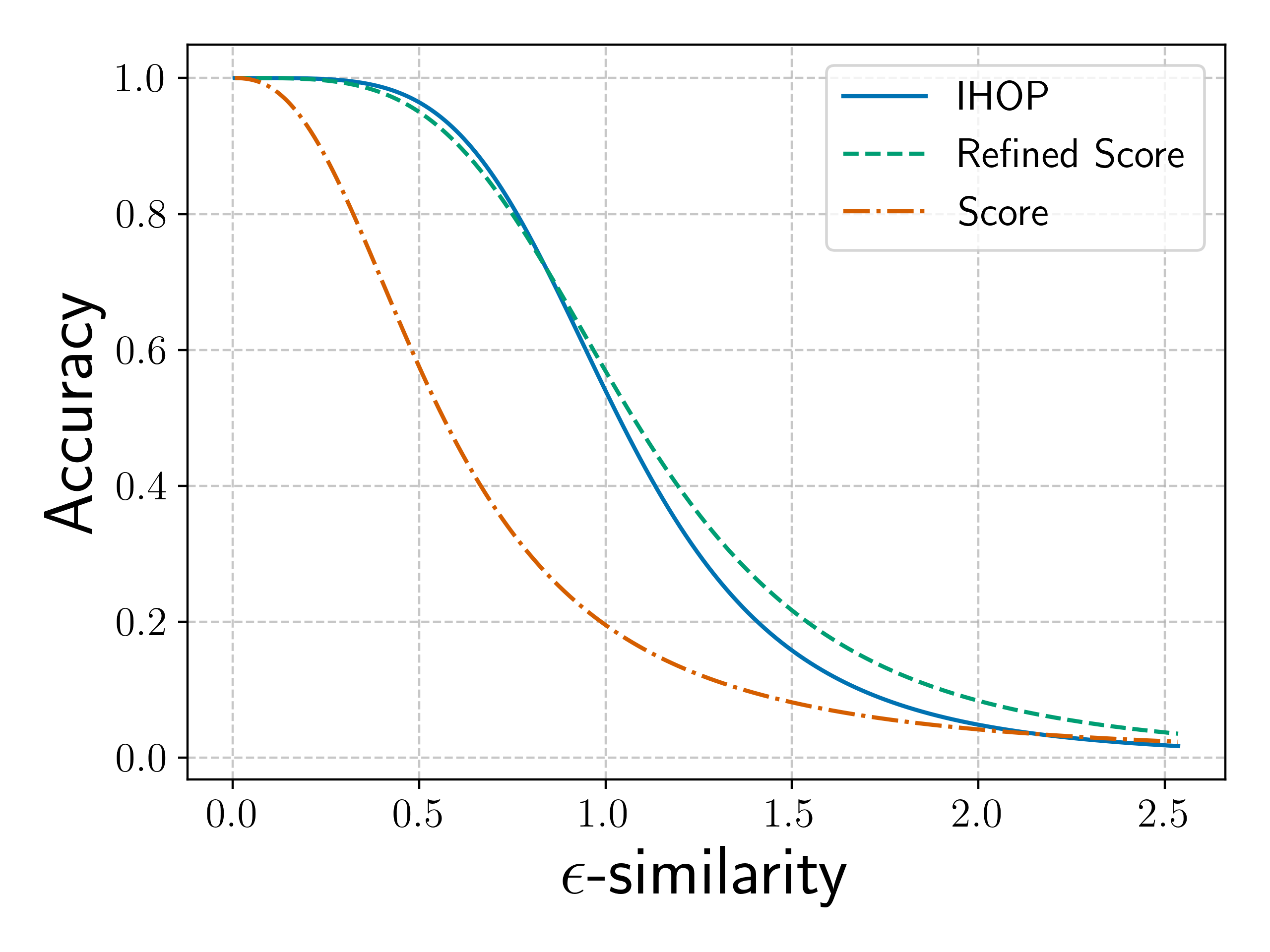}
    \caption{Comparison of the average accuracy functions of the Score, Refined Score, and IHOP attacks ($m=500$)}
    \label{fig:comp_atk}
\end{figure}

\paragraph{Attack mitigation efficiency}
Finally, the comparison framework can be extended to attack mitigation comparison.
The idea would be to compare accuracy functions obtained with and without countermeasures.
The impact of the mitigation on the accuracy function would quantify its efficiency. 

\section{Making the ``similar-data'' assumption unlikely}
\label{sec:dataset_size_acc}

Since attack papers require the attacker to have a ``sufficiently similar'' dataset, one may expect that only the attacker has an influence on the ``similar-data'' assumption.
Indeed, this formulation may imply that the assumption only depends on the quality of the dataset found by the attacker.
However, this intuition is untrue: some settings in the SSE deployment can make the ``similar-data'' assumption unlikely.
In particular, this section shows how to set a maximum index size such that it is highly unlikely to find a ``sufficiently similar'' dataset.

\subsection{Dataset size, an essential condition to attack success}
\label{subsec:bounded_acc}
Section \ref{subsec:size_similarity} have explored the influences of dataset sizes on the similarity.
A key take-away of this analysis is the similarity threshold induced by a fixed dataset size.
Indeed, to reach perfect similarity (with two independent datasets), one needs both an attacker's and an indexed dataset of infinite size.
If the indexed dataset has bounded size, the similarity is unlikely to go below a certain threshold (\emph{even with an infinite-sized attacker's dataset}).

Since the accuracy is correlated to the similarity (see Section \ref{sec:atk_analysis_comp}), we can deduce that: if the indexed dataset has bounded size, the attack accuracy is unlikely to go below a certain threshold.
Based on this intuition, we propose the following (almost) obvious theorem linking dataset sizes and attack accuracy:

\begin{theorem}
    For any distribution $\mathcal{D}$ (with  $D_\text{ind}$ and $D_\text{atk}$ drawn from $\mathcal{D}$), for any keyword universe $\mathcal{W}$, there exists a maximum index size $n_\text{max}\in \mathbb{N}$ such that:

    For any $n_\text{atk}\in \mathbb{N}$, if $n_\text{ind} \le n_\text{max}$, the average accuracy of any passive query-recovery attack is lower than or equal to $\frac{1}{m}$.
    \label{th:bounded-acc}
\end{theorem}
\begin{proof}
    We show that with $n_\text{max} = 0$, the statement in the theorem is always satisfied: the index contains no document, so all the queries return an empty set.
    Hence, an attacker can only randomly guess the queried keyword.
    The average accuracy is then $\frac{1}{m}$ (with $m$ the number of keywords in $\mathcal{W}$).
\end{proof}

Our proof relies on a trivial case (i.e., $n_\text{max}=0$).
Many data distributions could verify the theorem with a strictly positive $n_\text{max}$.
However, some distributions cannot.
Let us consider the keyword universe $\mathcal{W}=\{w_1, w_2\}$ and a distribution $\mathcal{D}$ such that $\mathbb{P}_\mathcal{D}(w_1) =1$ and $\mathbb{P}_\mathcal{D}(w_2) =0$.
If $D_\text{ind}$ contains even a single document, the query leakage would be sufficient to recover the queries perfectly: if the result is empty, it is $w_2$; otherwise, it is $w_1$.

Moreover, there exist other distributions for which Theorem \ref{th:bounded-acc} holds for all $n_\text{max} \in \mathbb{N}$.
For example, if each indexed document contains all the keywords from $\mathcal{W}$, each query returns all the documents, so the leakage is useless (i.e., no attack is better than a random guess).
Even with an infinite-sized indexed dataset, the best average attack accuracy would be $\frac{1}{m}$.

In practice, real-world distibutions would have an $n_\text{max}$ between these two edge cases.
Our goal is now to estimate this value for these distributions.

\subsection{Estimating the maximum index size}
\label{subsec:estimation_protocol}
To estimate this maximum index size, we want an upper bound on the attack accuracy and then use it to deduce the maximum index size.
Let $f_\text{MaxAcc}(n_\text{ind}, n_\text{atk})$ be the accuracy upper bound in function of the dataset sizes $(n_\text{ind}, n_\text{atk})$.
The maximum index size $n_\text{max}$ is such that $\lim_{n_\text{atk} \rightarrow \infty} f_\text{MaxAcc}(n_\text{max},n_\text{atk}) \le \frac{1}{m}$.
We can generalize this formula and consider a generic maximum accuracy threshold $\beta_\text{max}$ (in $\left[\frac{1}{m},1\right]$) instead of $\frac{1}{m}$.

Ideally, we could deduce analytically an accuracy upper bound for all attacks.
Unfortunately, a theoretical bound could be non-informative.
For example, Section \ref{subsec:bounded_acc} highlights a distribution for which the attack accuracy is always perfect.
Hence, a purely theoretical bound could be too loose due to these edge cases.

To avoid these issues, we will compute this accuracy upper bound using statistical estimation.
This estimated upper bound is not be absolute, but simply statistical: it is not impossible to be above this threshold; it is only unlikely.
This statistical approach can provide tight and informative bounds \emph{for a given attack and a given use case} (represented by a sample dataset).

\paragraph{Estimation protocol}
Section \ref{sec:atk_analysis_comp} presents the estimation of an average accuracy function.
We can extend this approach to estimate our maximum accuracy function $f_\text{MaxAcc}$, but we must make two changes to the estimation process.

First, we replace $\epsilon$ by the value $\sqrt{\frac{1}{n_\text{ind}} + \frac{1}{n_\text{atk}}}$ in the estimation process.
This replacement is motivated by the linear relationship between $\sqrt{\frac{1}{n_\text{ind}} + \frac{1}{n_\text{atk}}}$ and the similarity metric $\epsilon$ discovered in Section \ref{subsec:size_similarity}.
Thanks to this linear relationship, we can replace $\epsilon$ from the average accuracy estimation to build a regression model linking the accuracy to $(n_\text{ind}, n_\text{atk})$ (instead of the accuracy to $\epsilon$).

Second, we want an accuracy upper bound (not an average accuracy), so we cannot use linear regression.
Instead, our estimation protocol relies on quantile regression \cite{hao_quantile_2007,koenker_regression_1978,koenker_quantile_2001}.
In a quantile regression, the resulting estimated function describes the quantile of a data distribution\footnote{The quantile of level $\alpha$ for distribution $Y$ is defined as follows: $Q_Y(\alpha) = \inf \{y: F_Y(y) \ge \alpha\}$, with $F_Y$ the cumulative distribution function of $Y$.} instead of the average case (as linear regression does).
A quantile regression computes the parameters $(b,a)$ such that $Q_Y(\alpha) = b\cdot X + a$, for $(X,Y)$ two data distributions and $\alpha$ a quantile level.
We refer to \cite{hao_quantile_2007,koenker_regression_1978,koenker_quantile_2001} for details about the computation.
This quantile regression is ideal for representing a maximum accuracy bound with high probability.

To sum up, our quantile regression estimates (for a high quantile $\alpha$) the parameter $(a, b)$ such that $Q_{\logit(\text{Acc})}(\alpha) = b\log(\sqrt{\frac{1}{n_\text{ind}} + \frac{1}{n_\text{atk}}}) + a$.
Then, we deduce our estimated upper bound $\widehat{f}_\text{MaxAcc} = \expit(b\cdot\log(\sqrt{\frac{1}{n_\text{ind}} + \frac{1}{n_\text{atk}}}) + a)$.
Now that we have a satisfying estimation protocol, we only need to generate simulation results (using a sample dataset representative of the use case) with varying dataset sizes and compute the quantile regression.

For the quantile choice, the thresholds 0.1\%, 1\%, and 5\% are recurrent in statistics.
The smaller the quantile is, the more conservative the estimation is.

\begin{figure*}[t]
    \centering
    \begin{subfigure}{0.45\textwidth}
        \centering
        \includegraphics[width=\linewidth]{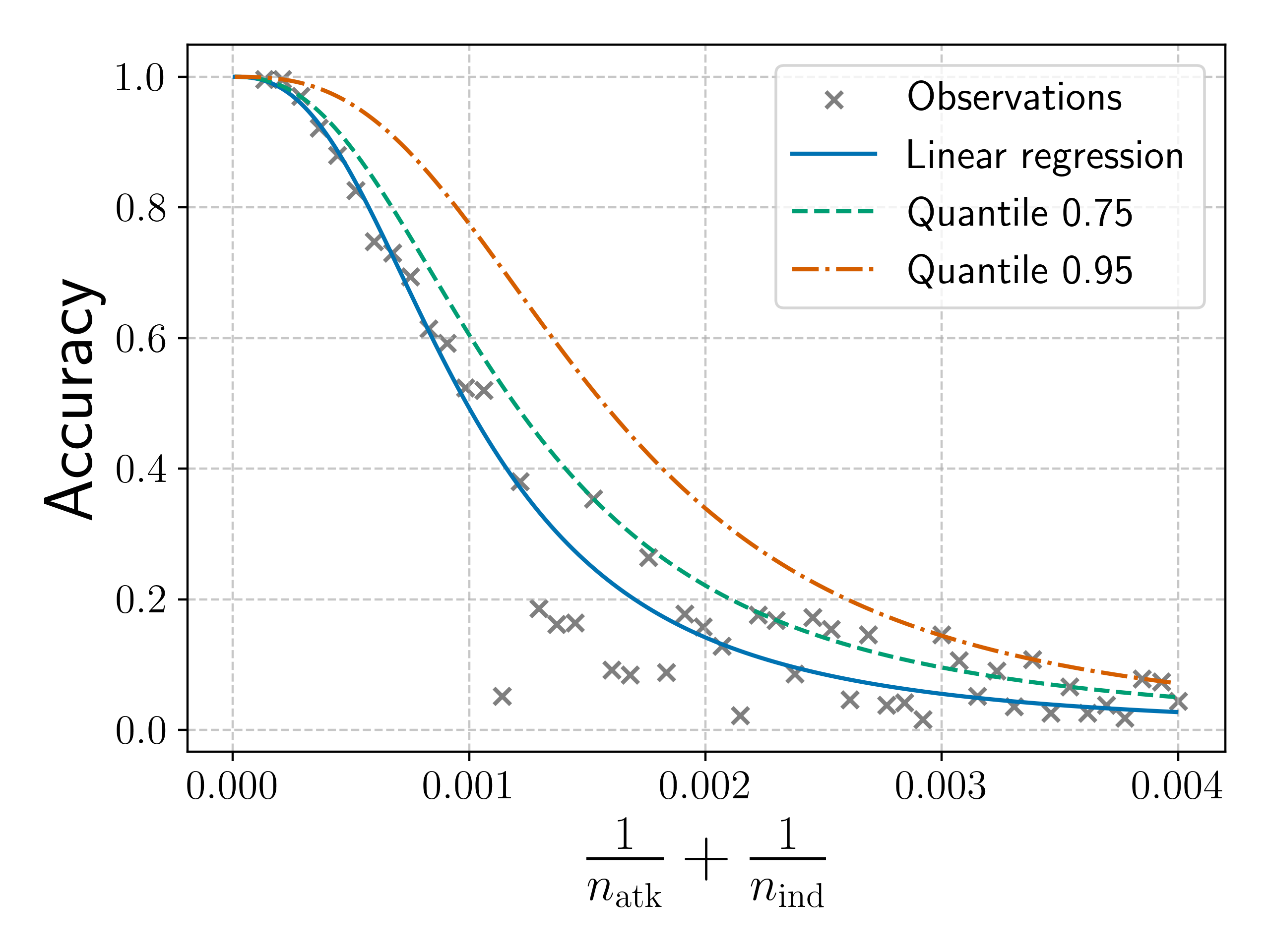}
        \caption{For the IHOP attack}
        \label{fig:ihop_acc_upper_bound}
    \end{subfigure}
    \begin{subfigure}{0.45\textwidth}
        \includegraphics[width=\linewidth]{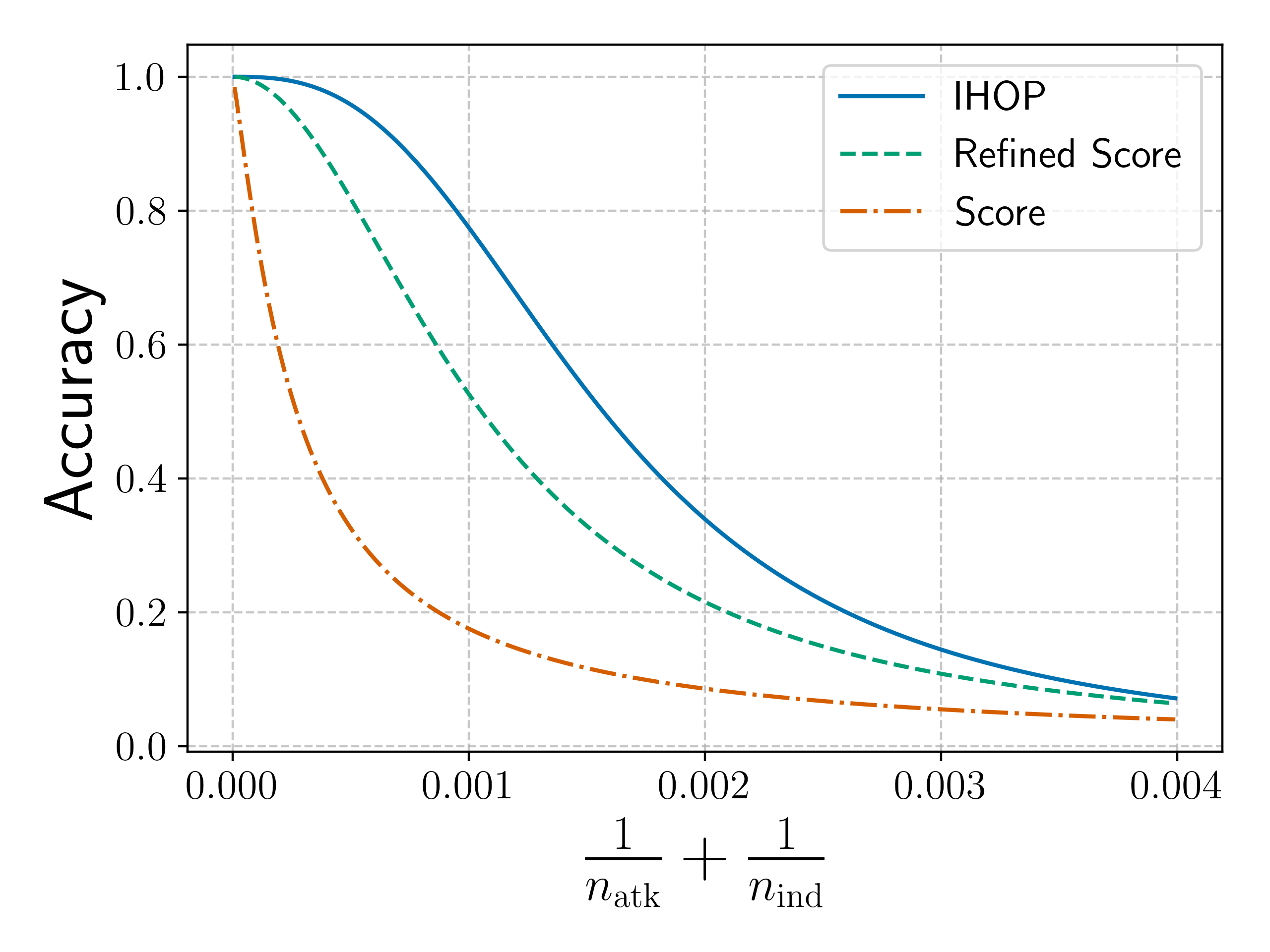}
        \caption{Comparison of three attacks \cite{damie_highly_2021,oya_ihop_2022} with the quantile 0.95}
        \label{fig:comp_acc_upper_bound}
    \end{subfigure}
    \caption{Estimated accuracy upper bounds (500 keywords / Dataset: Enron)}
    \label{fig:acc_upper_bound}
\end{figure*}

\paragraph{Limitations}
This statistical approach provides a tight upper bound, but it has a drawback: the upper bound holds for a specific attack on a specific dataset.
Thus, the maximum index size must be recalculated if a new attack is proposed, or if the use case changes.

While it is limiting, this observation holds for many attack mitigations.
For example, index padding \cite{cash_leakage-abuse_2015,xu_hardening_2019} consists in adding false-positive results to the index.
These works do not present a theoretical manner to set the padding parameters.
Hence, like our maximum index size evaluation, these mitigations rely on cryptanalysis.

Anyway, our goal is not to present a perfect mitigation.
We want to show that practitioners can limit the similar-data assumption using a simple technique.

\paragraph{A conservative estimation}
As our statistical approach provides an attack mitigation usable by practitioners (i.e., setting a maximum index size), we build our estimation protocol to be highly conservative: a real-world attacker is unlikely to reach the accuracy bound.
The quantile regression is the main component contributing to this conservative estimation since it provides an upper bound contrary to the average function of the linear regression.
Moreover, we also obtain attack results assuming advantageous attack conditions: \emph{(1)} the attacker knows the whole keyword universe, \emph{(2)} the attacker observed all possible queries, \emph{(3)} we use uniform splitting for dataset generation (Appendix \ref{app:uniform_split} proves that it simulates best-case scenario for the attacker).
\emph{In short, real-world attackers cannot benefit from better conditions.}

\paragraph{Divergence from Theorem \ref{th:bounded-acc}}
Our (conservative) statistical estimation bounds the accuracy with ``high probability'', while Theorem \ref{th:bounded-acc} bounds the average case.
Our bound is then stronger than the one described in Theorem \ref{th:bounded-acc}.
The average case considered in Theorem \ref{th:bounded-acc} provides a shorter formulation, but we could prove a variant providing an upper bound with high probability using a similar proof.

\subsection{Example on Enron dataset}
\label{subsec:example_enron}
Figure \ref{fig:ihop_acc_upper_bound} shows the estimated accuracy upper bound for the IHOP attack on the Enron dataset.
We observe a smooth and coherent upper bound (with regard to the attack results).

We reproduced this estimation protocol for the score and refined score attacks \cite{damie_highly_2021} and plot the 0.95 quantiles of all three attacks in Figure \ref{fig:comp_acc_upper_bound}.
Surprisingly, Figure \ref{fig:comp_acc_upper_bound} shows that the IHOP accuracy upper bound is always better than the Refined Score attack.
Indeed, Figure \ref{fig:comp_atk} of Section \ref{subsec:comp} showed that the Refined Score attack was \textbf{in average} better than the Refined Score attack on poorly similar datasets.
These observations mean that IHOP occasionally reaches very high accuracies, while refined score results have a smaller variance and higher mean.
Despite these minor differences, both attacks are accurate in similar scenarios, but IHOP is better because it requires no known queries.


\begin{figure}
    \begin{subfigure}{0.45\linewidth}
        \centering
        \includegraphics[width=\linewidth]{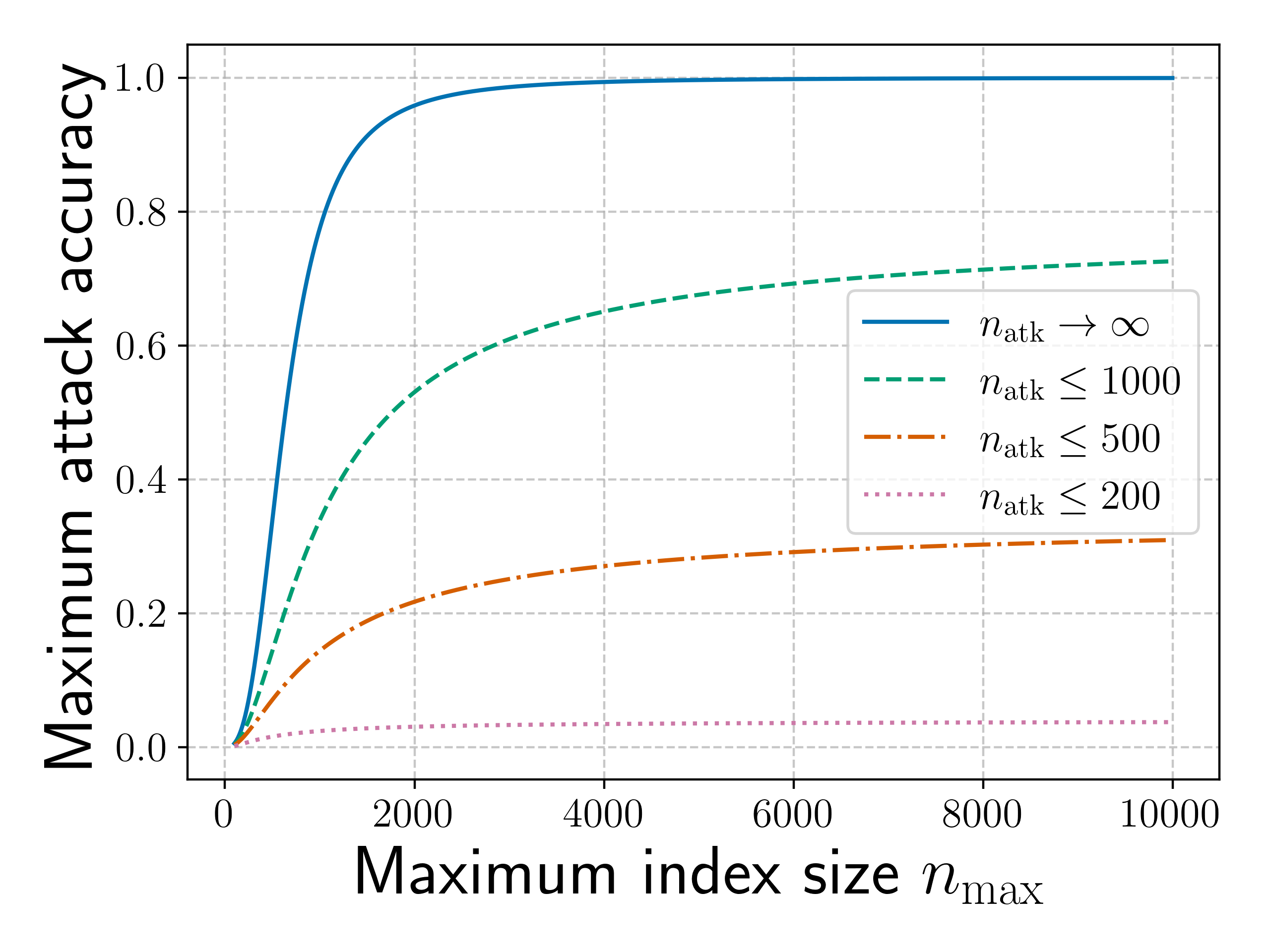}
        \caption{Normal SSE index}
        \label{fig:max_index_size}
    \end{subfigure}
    \begin{subfigure}{0.45\linewidth}
        \centering
        \includegraphics[width=\linewidth]{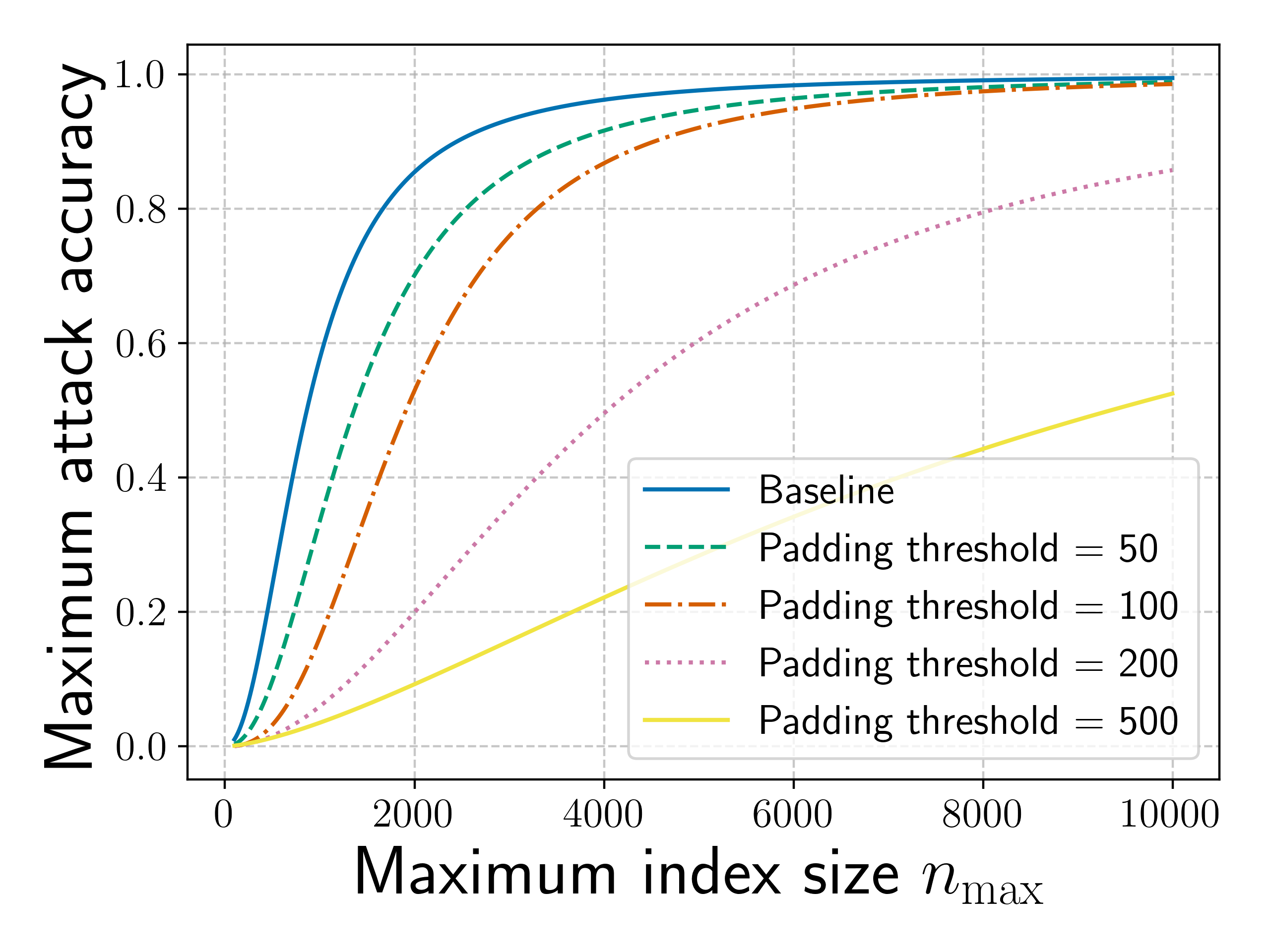}
        \caption{Padded SSE index \cite{cash_leakage-abuse_2015}}
        \label{fig:max_index_size_with_mitigation}
    \end{subfigure}
    \caption{Maximum attack accuracy in function of the maximum index size $n_\text{max}$ (quantile regression with $\alpha = 0.95$ / Dataset: Enron / Attack: IHOP \cite{oya_ihop_2022})}
\end{figure}

As IHOP has the highest accuracy upper bound, we use this attack to deduce the maximum index size.
Figure \ref{fig:max_index_size} shows the maximum index size for varying maximum accuracies $\beta_\text{max}$.
The top curve (in blue) depicts the most conservative setup: adversary with infinite-sized attacker's dataset.

For example, a maximum accuracy of 5\% requires a maximum size of 218 documents per index.
While this constraint could be too strict for some applications, some measures could allow satisfying it in practice.
Let us assume we store emails (like those from Enron) using SSE, we could (1) enforce a regular cleaning of mailboxes or (2) instantiate a separate secure index for each folder.

To increase this maximum size, we propose two solutions.
On the one hand, we could relax our assumptions and consider a finite-sized attacker's dataset.
The three bottom curves of Figure \ref{fig:max_index_size} shows the estimation with three different bounds on the attacker's dataset size: 1000, 500, 200.
On the other hand, we could apply other attack mitigations such as padding \cite{cash_leakage-abuse_2015,xu_hardening_2019} and compute the maximum index size for a padded index (see Figure \ref{fig:max_index_size_with_mitigation}).

Our results shows that SSE practitioners can have a direct influence on the ``similar-data'' assumption.
While this opens new directions in SSE attack mitigation (i.e., via maximum index size), further theoretical analysis is necessary to build a complete security framework supporting such statistical approaches.

\section{Do data breaches provide ``sufficiently similar'' data?}
\label{sec:dataset_generation}
Attack papers \cite{blackstone_revisiting_2020,damie_highly_2021} often motivate the existence of an attacker's dataset based on data breaches.
SSE attackers could build their dataset using previously leaked documents, but do these data breaches provide ``sufficiently similar'' data?

\paragraph{Simulating a data breach}
Inference-attack papers \cite{damie_highly_2021,oya_ihop_2022} generate indexed and attacker's datasets by splitting a real-world dataset (e.g., Apache emails) into two disjoint sets.
By default, all papers split the dataset \emph{uniformly at random}.
We propose to compare this splitting method to a more realistic alternative: the temporal split.

The temporal split consists in splitting the dataset based on the document timestamp.
With email datasets such as Apache, the attacker knows all the emails before a given date, and the index contains all emails sent after this date.
This splitting simulates nicely a data breach: the split timestamp corresponds to the breach date; the attacker has access to all data before the breach.

\begin{table}[t]
    \caption{Accuracy of the refined score attack \cite{damie_highly_2021} and $\epsilon$-similarity on the Apache dataset (with $\left|\mathcal{W}\right|=1K$)}
    \begin{subtable}{0.45\textwidth}
        \centering
        \begin{tabular}{|c|c|c|c|c|}\hline
            Year                  & 2003 & 2005 & 2007 & 2009 \\\hline
            $\epsilon$-similarity & 4.90 & 4.07 & 4.52 & 4.62 \\\hline
            Attack acc. (\%)      & 0.70 & 2.81 & 2.81 & 1.75 \\\hline
        \end{tabular}

        \caption{Temporal split}
        \label{tab:acc-temporal-split}
    \end{subtable}
    \begin{subtable}{0.45\textwidth}
        \centering
        \begin{tabular}{|c|c|c|c|}\hline
                                  & Average & Min.  & Max.  \\\hline
            $\epsilon$-similarity & 0.62    & 0.59  & 0.71  \\ \hline
            Attack acc. (\%)      & 98.20   & 94.39 & 99.65 \\ \hline
        \end{tabular}
        \caption{Random split (100 repetitions).}
        \label{tab:acc-rand-split}
    \end{subtable}
\end{table}

\paragraph{Accuracy discrepancy}
Table \ref{tab:acc-temporal-split} presents the accuracy (of the refined score attack \cite{damie_highly_2021}) and the $\epsilon$-similarity obtained using the temporal split with four different splitting dates.
Table \ref{tab:acc-rand-split} presents the accuracy and the $\epsilon$-similarity obtained over 100 repetitions of the uniform split.
With the temporal split, the attack accuracy is always below 5\%, while it is always above 94\% with the uniform split.
The $\epsilon$-similarity results confirm this observation since the $\epsilon$ is much higher with the temporal split.
The low accuracy with the temporal split is not due to the dataset size since the split on the year 2007 provides sizes equivalent to those used for the uniform split (i.e., around 25K documents per dataset).
Therefore, the uniform split used in all attack papers can simulate an overly powerful attacker compared to a more realistic method.

The temporal split decreases the similarity (and the accuracy) compared to the uniform split because the distribution of the Apache dataset shifts over time.
Appendix \ref{app:splitting-equality-coprob} rigorously identifies the shift using statistical tests.
Note that if the Apache dataset had no distribution shift over time, the uniform and the temporal split should provide the same results.

Other datasets (\emph{without} a distribution shift) would have no accuracy difference between the temporal and uniform split.
However, distribution shifts (also referred to as ``concept drift'' in the literature) are common in real-world datasets; making them a recurrent concern in ML research \cite{chen_mandoline_2021,tsymbal_problem_2004,japkowicz_overview_2016}.
\textbf{Hence, this phenomenon could hinder attackers relying on data breaches from building similar datasets.}
This observation weakens the recurrent claim of attack papers \cite{blackstone_revisiting_2020,damie_highly_2021,islam_access_2012} motivating the existence of a similar dataset thanks to possible data breaches.
The existence of a ``sufficiently similar'' attacker's dataset could be a stronger assumption than expected.

\paragraph{What is the correct method for research works?}
Even if the uniform split simulates an overly powerful attacker, researchers should keep using it to obtain a conservative attack analysis: a real-world attacker would not have better conditions since the uniform splitting simulates a best-case scenario for the attacker (as detailed in Appendix \ref{app:uniform_split}).

Attack papers could also provide results based on the temporal split to evaluate the attacks in more realistic setups.
Such analysis would incentivize novel attacks leveraging a potential distribution shift to improve their accuracy.

\section*{Conclusion}

Our work provided a novel understanding of the ``similar-data'' assumption used in inference attacks against SSE.
From the factors influencing the attack accuracy to a mechanism to make the assumption unlikely, we raised several concerns regarding the practicality of this assumption in existing attacks.
Overall, our results all points at the same direction: the similar-data assumption can be hard to fulfil.
In particular, we designed multiple statistical tools usable on any future attack to analyze their performances and their sensitivity to data similarity.

Our results raised new questions about the security of SSE schemes.
While the literature focused on leakage minimization to guarantee security, our results (especially the maximum index size) opens the way to new approaches combining statistics with cryptanalysis to estimate real-world risks.

\paragraph{Future works}
We focused on the most attacked SSE setting, but our statistical approach can be extended to any SSE attack assuming an auxiliary data knowledge (e.g., attacks using query frequency).
Such extensions require additional works to adapt our mathematical model and statistical tools.

\begin{credits}
    \subsubsection{\ackname}
    This work was supported by the ANR project ANR-20-CE23-0013 `PMR' and the Netherlands Organization for Scientific Research (De Nederlandse Organisatie voor Wetenschappelijk Onderzoek) under NWO:SHARE project [CS.011].

\end{credits}
%
%
%
\bibliographystyle{splncs04}
\bibliography{main}

\begin{thebibliography}{10}
\providecommand{\url}[1]{\texttt{#1}}
\providecommand{\urlprefix}{URL }
\providecommand{\doi}[1]{https://doi.org/#1}

\bibitem{blackstone_revisiting_2020}
Blackstone, L., Kamara, S., Moataz, T.: Revisiting {Leakage} {Abuse} {Attacks}.
  In: Network and {Distributed} {System} {Security} {Symposium} ({NDSS}) (2020)

\bibitem{bland_multiple_1995}
Bland, J.M., Altman, D.G.: Multiple significance tests: the {Bonferroni}
  method. BMJ  \textbf{310}(6973), ~170 (Jan 1995).
  \doi{10.1136/bmj.310.6973.170}

\bibitem{boldyreva_understanding_2024}
Boldyreva, A., Gui, Z., Warinschi, B.: Understanding {Leakage} in {Searchable}
  {Encryption}: a {Quantitative} {Approach}. Proceedings on Privacy Enhancing
  Technologies  \textbf{2024}(4),  503--524 (Oct 2024).
  \doi{10.56553/popets-2024-0127}

\bibitem{bost_sophos_2016}
Bost, R.: Sophos: {Forward} {Secure} {Searchable} {Encryption}. In: Proceedings
  of the 2016 {ACM} {SIGSAC} {Conference} on {Computer} and {Communications}
  {Security}. pp. 1143--1154. {CCS} '16, Association for Computing Machinery
  (Oct 2016). \doi{10.1145/2976749.2978303}

\bibitem{bost_thwarting_2017}
Bost, R., Fouque, P.A.: Thwarting {Leakage} {Abuse} {Attacks} against
  {Searchable} {Encryption} -- {A} {Formal} {Approach} and {Applications} to
  {Database} {Padding} (2017)

\bibitem{bost_forward_2017}
Bost, R., Minaud, B., Ohrimenko, O.: Forward and {Backward} {Private}
  {Searchable} {Encryption} from {Constrained} {Cryptographic} {Primitives}.
  In: Proceedings of the 2017 {ACM} {SIGSAC} {Conference} on {Computer} and
  {Communications} {Security}. pp. 1465--1482. ACM (Oct 2017).
  \doi{10.1145/3133956.3133980}

\bibitem{box_science_1976}
Box, G.E.P.: Science and {Statistics}. Journal of the American Statistical
  Association  \textbf{71}(356),  791--799 (Dec 1976).
  \doi{10.1080/01621459.1976.10480949}

\bibitem{cash_leakage-abuse_2015}
Cash, D., Grubbs, P., Perry, J., Ristenpart, T.: Leakage-{Abuse} {Attacks}
  {Against} {Searchable} {Encryption}. In: Proceedings of the 22nd {ACM}
  {SIGSAC} {Conference} on {Computer} and {Communications} {Security}. pp.
  668--679. {CCS} ’15, Association for Computing Machinery (2015).
  \doi{10.1145/2810103.2813700}

\bibitem{cash_dynamic_2014}
Cash, D., Jaeger, J., Jarecki, S., Jutla, C., Krawczyk, H., Roşu, M.C.,
  Steiner, M.: Dynamic {Searchable} {Encryption} in {Very}-{Large} {Databases}:
  {Data} {Structures} and {Implementation}. In: Proceedings 2014 {Network} and
  {Distributed} {System} {Security} {Symposium}. Internet Society (2014).
  \doi{10.14722/ndss.2014.23264}

\bibitem{chase_structured_2010}
Chase, M., Kamara, S.: Structured {Encryption} and {Controlled} {Disclosure}.
  In: Abe, M. (ed.) Advances in {Cryptology} - {ASIACRYPT} 2010. pp. 577--594.
  Lecture {Notes} in {Computer} {Science}, Springer (2010).
  \doi{10.1007/978-3-642-17373-8_33}

\bibitem{chen_mandoline_2021}
Chen, M., Goel, K., Sohoni, N.S., Poms, F., Fatahalian, K., Re, C.: Mandoline:
  {Model} {Evaluation} under {Distribution} {Shift}. In: Proceedings of the
  38th {International} {Conference} on {Machine} {Learning}. pp. 1617--1629.
  PMLR (Jul 2021), iSSN: 2640-3498

\bibitem{clogg_statistical_1995}
Clogg, C.C., Petkova, E., Haritou, A.: Statistical {Methods} for {Comparing}
  {Regression} {Coefficients} {Between} {Models}. American Journal of Sociology
   \textbf{100}(5),  1261--1293 (Mar 1995). \doi{10.1086/230638}

\bibitem{cramer_mathematical_1999}
Cramér, H.: Mathematical {Methods} of {Statistics}. Princeton University Press
  (Apr 1999)

\bibitem{curtmola_searchable_2006}
Curtmola, R., Garay, J., Kamara, S., Ostrovsky, R.: Searchable {Symmetric}
  {Encryption}: {Improved} {Definitions} and {Efficient} {Constructions}. In:
  Proceedings of the 13th {ACM} {Conference} on {Computer} and {Communications}
  {Security}. pp. 79--88. {CCS} '06, Association for Computing Machinery
  (2006). \doi{10.1145/1180405.1180417}

\bibitem{damie_highly_2021}
Damie, M., Hahn, F., Peter, A.: A {Highly} {Accurate} {Query}-{Recovery}
  {Attack} against {Searchable} {Encryption} using {Non}-{Indexed} {Documents}.
  In: 30th {USENIX} {Security} {Symposium} ({USENIX} {Security} 21). pp.
  143--160 (2021)

\bibitem{dijkslag_passive_2022}
Dijkslag, M., Damie, M., Hahn, F., Peter, A.: Passive {Query}-{Recovery}
  {Attack} {Against} {Secure} {Conjunctive} {Keyword} {Search} {Schemes}. In:
  Ateniese, G., Venturi, D. (eds.) Applied {Cryptography} and {Network}
  {Security}. pp. 126--146. Lecture {Notes} in {Computer} {Science}, Springer
  International Publishing (2022). \doi{10.1007/978-3-031-09234-3_7}

\bibitem{dittert_too_2023}
Dittert, D., Schneider, T., Treiber, A.: Too {Close} for {Comfort}? {Measuring}
  {Success} of {Sampled}-{Data} {Leakage} {Attacks} {Against} {Encrypted}
  {Search}. In: Cloud {Computing} {Security} {Workshop} ({CCSW} ’23) (2023)

\bibitem{everitt_cambridge_2010}
Everitt, B.S., Skrondal, A.: The {Cambridge} dictionary of statistics.
  Cambridge University Press (2010)

\bibitem{falzon_full_2020}
Falzon, F., Markatou, E.A., Akshima, Cash, D., Rivkin, A., Stern, J., Tamassia,
  R.: Full {Database} {Reconstruction} in {Two} {Dimensions}. In: Proceedings
  of the 2020 {ACM} {SIGSAC} {Conference} on {Computer} and {Communications}
  {Security}. pp. 443--460. {CCS} '20, Association for Computing Machinery, New
  York, NY, USA (Nov 2020). \doi{10.1145/3372297.3417275}

\bibitem{ghareh_chamani_new_2018}
Ghareh~Chamani, J., Papadopoulos, D., Papamanthou, C., Jalili, R.: New
  {Constructions} for {Forward} and {Backward} {Private} {Symmetric}
  {Searchable} {Encryption}. In: Proceedings of the 2018 {ACM} {SIGSAC}
  {Conference} on {Computer} and {Communications} {Security}. pp. 1038--1055.
  ACM (Oct 2018). \doi{10.1145/3243734.3243833}

\bibitem{griliches1983handbook}
Griliches, Z., Heckman, J.J., Intriligator, M.D., Engle, R.F., Leamer, E.E.,
  McFadden, D.: Handbook of econometrics. Elsevier (1983)

\bibitem{grubbs_pump_2018}
Grubbs, P., Lacharité, M.S., Minaud, B., Paterson, K.G.: Pump up the volume:
  {Practical} database reconstruction from volume leakage on range queries. In:
  Proceedings of the 2018 {ACM} {SIGSAC} {Conference} on {Computer} and
  {Communications} {Security}. pp. 315--331 (2018)

\bibitem{gui_rethinking_2023}
Gui, Z., Paterson, K.G., Patranabis, S.: Rethinking {Searchable} {Symmetric}
  {Encryption}. In: 2023 {IEEE} {Symposium} on {Security} and {Privacy} ({SP}).
  p.~44 (2023)

\bibitem{hadar_rules_1969}
Hadar, J., Russell, W.R.: Rules for ordering uncertain prospects. The American
  economic review  \textbf{59}(1),  25--34 (1969)

\bibitem{hao_quantile_2007}
Hao, L., Naiman, D.Q.: Quantile regression. Sage (2007)

\bibitem{ho_similar_2024}
Ho, B., Chen, H., Shi, Z., Liang, K.: Similar {Data} is {Powerful}: {Enhancing}
  {Inference} {Attacks} on {SSE} with {Volume} {Leakages}. In: Garcia-Alfaro,
  J., Kozik, R., Choraś, M., Katsikas, S. (eds.) Computer {Security} –
  {ESORICS} 2024. pp. 105--126. Springer Nature Switzerland, Cham (2024).
  \doi{10.1007/978-3-031-70903-6_6}

\bibitem{holm1979simple}
Holm, S.: A simple sequentially rejective multiple test procedure. Scandinavian
  journal of statistics pp. 65--70 (1979)

\bibitem{hsu1996multiple}
Hsu, J.: Multiple comparisons: theory and methods. CRC Press (1996)

\bibitem{islam_access_2012}
Islam, M.S., Kuzu, M., Kantarcioglu, M.: Access pattern disclosure on
  searchable encryption: {Ramification}, attack and mitigation. In: Network and
  {Distributed} {System} {Security} {Symposium} ({NDSS}) (2012)

\bibitem{kamara_sok_2022}
Kamara, S., Kati, A., Moataz, T., Schneider, T., Treiber, A., Yonli, M.: {SoK}:
  {Cryptanalysis} of {Encrypted} {Search} with {LEAKER} - {A} framework for
  {LEakage} {AttacK} {Evaluation} on {Real}-world data. In: {IEEE} {European}
  {Symposium} on {Security} and {Privacy} ({EuroS}\&{P}'22). IEEE (Jun 2022)

\bibitem{kamara_bayesian_2023}
Kamara, S., Moataz, T.: Bayesian {Leakage} {Analysis}: {A} {Framework} for
  {Analyzing} {Leakage} in {Encrypted} {Search} (2023)

\bibitem{kamara_structured_2018}
Kamara, S., Moataz, T., Ohrimenko, O.: Structured {Encryption} and {Leakage}
  {Suppression}. In: Shacham, H., Boldyreva, A. (eds.) Advances in {Cryptology}
  – {CRYPTO} 2018. pp. 339--370. Springer International Publishing, Cham
  (2018). \doi{10.1007/978-3-319-96884-1_12}

\bibitem{kanji_100_2006}
Kanji, G.K.: 100 statistical tests. Sage Publications, 3rd ed edn. (2006)

\bibitem{kellaris_generic_2016}
Kellaris, G., Kollios, G., Nissim, K., O'neill, A.: Generic attacks on secure
  outsourced databases. In: Proceedings of the 2016 {ACM} {SIGSAC} {Conference}
  on {Computer} and {Communications} {Security}. pp. 1329--1340 (2016)

\bibitem{klimt2004introducing}
Klimt, B., Yang, Y.: Introducing the enron corpus. In: CEAS (2004)

\bibitem{koenker_regression_1978}
Koenker, R., Bassett, G.: Regression quantiles. Econometrica pp. 33--50 (1978)

\bibitem{koenker_quantile_2001}
Koenker, R., Hallock, K.F.: Quantile {Regression}. Journal of Economic
  Perspectives  \textbf{15}(4),  143--156 (Dec 2001).
  \doi{10.1257/jep.15.4.143}

\bibitem{kornaropoulos_leakage_2022}
Kornaropoulos, E.M., Moyer, N., Papamanthou, C., Psomas, A.: Leakage
  {Inversion}: {Towards} {Quantifying} {Privacy} in {Searchable} {Encryption}.
  In: Proceedings of the 2022 {ACM} {SIGSAC} {Conference} on {Computer} and
  {Communications} {Security}. pp. 1829--1842. {CCS} '22, Association for
  Computing Machinery (Nov 2022). \doi{10.1145/3548606.3560593}

\bibitem{kornaropoulos_data_2019}
Kornaropoulos, E.M., Papamanthou, C., Tamassia, R.: Data {Recovery} on
  {Encrypted} {Databases} with k-{Nearest} {Neighbor} {Query} {Leakage}. In:
  2019 {IEEE} {Symposium} on {Security} and {Privacy} ({SP}). pp. 1033--1050
  (May 2019). \doi{10.1109/SP.2019.00015}

\bibitem{kornaropoulos_state_2020}
Kornaropoulos, E.M., Papamanthou, C., Tamassia, R.: The {State} of the
  {Uniform}: {Attacks} on {Encrypted} {Databases} {Beyond} the {Uniform}
  {Query} {Distribution}. In: 2020 {IEEE} {Symposium} on {Security} and
  {Privacy} ({SP}). pp. 1223--1240. IEEE, San Francisco, CA, USA (May 2020).
  \doi{10.1109/SP40000.2020.00029}

\bibitem{lacharite_improved_2018}
Lacharité, M.S., Minaud, B., Paterson, K.G.: Improved reconstruction attacks
  on encrypted data using range query leakage. In: 2018 {IEEE} {Symposium} on
  {Security} and {Privacy} ({SP}). pp. 297--314. IEEE (2018)

\bibitem{lambregts_volume_2022}
Lambregts, S., Chen, H., Ning, J., Liang, K.: {VAL}: {Volume} and {Access}
  {Pattern} {Leakage}-{Abuse} {Attack} with {Leaked} {Documents}. In: Atluri,
  V., Di~Pietro, R., Jensen, C.D., Meng, W. (eds.) Computer {Security} –
  {ESORICS} 2022. pp. 653--676. Lecture {Notes} in {Computer} {Science},
  Springer International Publishing (2022). \doi{10.1007/978-3-031-17140-6_32}

\bibitem{langhout_file-injection_2024}
Langhout, T.J., Chen, H., Liang, K.: File-{Injection} {Attacks}
  on {Searchable} {Encryption}, {Based} on {Binomial} {Structures}. In:
  Garcia-Alfaro, J., Kozik, R., Choraś, M., Katsikas, S. (eds.) Computer
  {Security} – {ESORICS} 2024. pp. 424--443. Springer Nature Switzerland,
  Cham (2024). \doi{10.1007/978-3-031-70896-1_21}

\bibitem{liu_search_2014}
Liu, C., Zhu, L., Wang, M., Tan, Y.A.: Search pattern leakage in searchable
  encryption: {Attacks} and new construction. Information Sciences
  \textbf{265},  176--188 (2014)

\bibitem{markatou_attacks_2023}
Markatou, E.A., Falzon, F., Espiritu, Z., Tamassia, R.: Attacks on {Encrypted}
  {Response}-{Hiding} {Range} {Search} {Schemes} in {Multiple} {Dimensions}.
  Proceedings on Privacy Enhancing Technologies  (2023)

\bibitem{markatou_reconstructing_2021}
Markatou, E.A., Falzon, F., Tamassia, R., Schor, W.: Reconstructing with
  {Less}: {Leakage} {Abuse} {Attacks} in {Two} {Dimensions}. In: Proceedings of
  the 2021 {ACM} {SIGSAC} {Conference} on {Computer} and {Communications}
  {Security}. pp. 2243--2261. {CCS} '21, Association for Computing Machinery,
  New York, NY, USA (Nov 2021). \doi{10.1145/3460120.3484552}

\bibitem{miller_simultaneous_2012}
Miller, R.G.J.: Simultaneous {Statistical} {Inference}. Springer Science \&
  Business Media (Dec 2012)

\bibitem{miller_explanation_2019}
Miller, T.: Explanation in artificial intelligence: {Insights} from the social
  sciences. Artificial Intelligence  \textbf{267},  1--38 (Feb 2019).
  \doi{10.1016/j.artint.2018.07.007}

\bibitem{ning_leap_2021leap}
Ning, J., Huang, X., Poh, G.S., Yuan, J., Li, Y., Weng, J., Deng, R.H.: {LEAP}:
  Leakage-abuse attack on efficiently deployable, efficiently searchable
  encryption with partially known dataset. In: Proceedings of the 2021 ACM
  SIGSAC Conference on Computer and Communications Security. pp. 2307--2320
  (2021)

\bibitem{ning_passive_2018}
Ning, J., Xu, J., Liang, K., Zhang, F., Chang, E.C.: Passive attacks against
  searchable encryption. IEEE Transactions on Information Forensics and
  Security  \textbf{14}(3),  789--802 (2018)

\bibitem{oya_hiding_2021}
Oya, S., Kerschbaum, F.: Hiding the {Access} {Pattern} is {Not} {Enough}:
  {Exploiting} {Search} {Pattern} {Leakage} in {Searchable} {Encryption}. In:
  30th {USENIX} {Security} {Symposium} ({USENIX} {Security} 21) (2021)

\bibitem{oya_ihop_2022}
Oya, S., Kerschbaum, F.: {IHOP}: {Improved} {Statistical} {Query} {Recovery}
  against {Searchable} {Symmetric} {Encryption} through {Quadratic}
  {Optimization}. In: 31st {USENIX} {Security} {Symposium} ({USENIX} {Security}
  22). pp. 2407--2424 (2022)

\bibitem{poddar_practical_2020}
Poddar, R., Wang, S., Lu, J., Popa, R.A.: Practical {Volume}-{Based} {Attacks}
  on {Encrypted} {Databases}. arXiv:2008.06627 [cs]  (Aug 2020)

\bibitem{porter1980algorithm}
Porter, M.F.: An algorithm for suffix stripping. Program  (1980)

\bibitem{pouliot_shadow_2016}
Pouliot, D., Wright, C.V.: The shadow nemesis: {Inference} attacks on
  efficiently deployable, efficiently searchable encryption. In: Proceedings of
  the 2016 {ACM} {SIGSAC} conference on computer and communications security.
  pp. 1341--1352 (2016)

\bibitem{rao_information_1945}
Rao, C.R.: Information and accuracy attainable in the estimation of statistical
  parameters. Bulletin of the Calcutta Mathematical Society  \textbf{37},
  81--91 (1945)

\bibitem{schler2006effects}
Schler, J., Koppel, M., Argamon, S., Pennebaker, J.W.: Effects of age and
  gender on blogging. In: AAAI spring symposium: Computational approaches to
  analyzing weblogs. vol.~6, pp. 199--205 (2006)

\bibitem{song_practical_2000}
Song, D.X., Wagner, D., Perrig, A.: Practical techniques for searches on
  encrypted data. In: Proceeding 2000 {IEEE} {Symposium} on {Security} and
  {Privacy}. {S}\&{P} 2000. pp. 44--55. IEEE (2000)

\bibitem{sun_practical_2021}
Sun, S.F., Steinfeld, R., Lai, S., Yuan, X., Sakzad, A., Liu, J., Nepal, S.,
  Gu, D.: Practical {Non}-{Interactive} {Searchable} {Encryption} with
  {Forward} and {Backward} {Privacy}. In: Proceedings 2021 {Network} and
  {Distributed} {System} {Security} {Symposium}. Internet Society (2021).
  \doi{10.14722/ndss.2021.24162}

\bibitem{sun_practical_2018}
Sun, S.F., Yuan, X., Liu, J.K., Steinfeld, R., Sakzad, A., Vo, V., Nepal, S.:
  Practical {Backward}-{Secure} {Searchable} {Encryption} from {Symmetric}
  {Puncturable} {Encryption}. In: Proceedings of the 2018 {ACM} {SIGSAC}
  {Conference} on {Computer} and {Communications} {Security}. pp. 763--780.
  {CCS} '18, Association for Computing Machinery (Oct 2018).
  \doi{10.1145/3243734.3243782}

\bibitem{tsymbal_problem_2004}
Tsymbal, A.: The problem of concept drift: definitions and related work.
  Computer Science Department, Trinity College Dublin  \textbf{106}(2), ~58
  (2004)

\bibitem{vo_shielddb_2023}
Vo, V., Yuan, X., Sun, S.F., Liu, J.K., Nepal, S., Wang, C.: {ShieldDB}: {An}
  {Encrypted} {Document} {Database} {With} {Padding} {Countermeasures}. IEEE
  Transactions on Knowledge and Data Engineering  \textbf{35}(4),  4236--4252
  (Apr 2023). \doi{10.1109/TKDE.2021.3126607}

\bibitem{weisstein_bonferroni_2004}
Weisstein, E.W.: Bonferroni correction. https://mathworld. wolfram. com/
  (2004)

\bibitem{xu_interpreting_2021}
Xu, L., Duan, H., Zhou, A., Yuan, X., Wang, C.: Interpreting and mitigating
  leakage-abuse attacks in searchable symmetric encryption. IEEE Transactions
  on Information Forensics and Security  (2021)

\bibitem{xu_hardening_2019}
Xu, L., Yuan, X., Wang, C., Wang, Q., Xu, C.: Hardening database padding for
  searchable encryption. In: {IEEE} {INFOCOM} 2019-{IEEE} {Conference} on
  {Computer} {Communications}. pp. 2503--2511. IEEE (2019)

\bibitem{xu_leakage-abuse_2023}
Xu, L., Zheng, L., Xu, C., Yuan, X., Wang, C.: Leakage-{Abuse} {Attacks}
  {Against} {Forward} and {Backward} {Private} {Searchable} {Symmetric}
  {Encryption}. In: Proceedings of the 2023 {ACM} {SIGSAC} {Conference} on
  {Computer} and {Communications} {Security}. {CCS}' 23, Association for
  Computing Machinery (2023)

\bibitem{zhang_inject_2024}
Zhang, M., Shi, Z., Chen, H., Liang, K.: Inject {Less}, {Recover} {More}:
  {Unlocking} the {Potential} of {Document} {Recovery} in {Injection} {Attacks}
  {Against} {SSE}. In: 2024 {IEEE} 37th {Computer} {Security} {Foundations}
  {Symposium} ({CSF}). pp. 311--323. IEEE Computer Society (Jul 2024).
  \doi{10.1109/CSF61375.2024.00029}

\bibitem{zhang_high_2023}
Zhang, X., Wang, W., Xu, P., Yang, L.T., Liang, K.: High {Recovery} with
  {Fewer} {Injections}: {Practical} {Binary} {Volumetric} {Injection} {Attacks}
  against {Dynamic} {Searchable} {Encryption}. In: 32nd {USENIX} {Security}
  {Symposium} ({USENIX} {Security} 23) (Feb 2023)

\bibitem{zhang_all_2016}
Zhang, Y., Katz, J., Papamanthou, C.: All your queries are belong to us: {The}
  power of file-injection attacks on searchable encryption. In: 25th {USENIX}
  {Security} {Symposium} ({USENIX} {Security} 16). pp. 707--720 (2016)

\bibitem{japkowicz_overview_2016}
Žliobaitė, I., Pechenizkiy, M., Gama, J.: An {Overview} of {Concept} {Drift}
  {Applications}. In: Japkowicz, N., Stefanowski, J. (eds.) Big {Data}
  {Analysis}: {New} {Algorithms} for a {New} {Society}, vol.~16, pp. 91--114.
  Springer International Publishing, Cham (2016).
  \doi{10.1007/978-3-319-26989-4_4}

\end{thebibliography}

\appendix

\section{Uniform dataset splitting, a favored attacker simulation}
\label{app:uniform_split}
This appendix shows that the classic attack simulation using uniform splitting for dataset generation creates the best-case scenario for the attacker.
Hence, any result obtained with uniform splitting is, on average, greater than or equal to the accuracy of a real-world attacker.
We prove it in three steps: \emph{(1)} uniform splitting produces dataset distributions with equal co-probabilities, \emph{(2)} equal co-probabilities lead to smaller $\epsilon$-similarity, \emph{(3)} smaller $\epsilon$ leads to higher accuracy.
This appendix focuses on proving analytically the second step, and we rely on auxiliary results for steps 1 and 3.

Appendix \ref{app:splitting-equality-coprob} covers the first step using a statistical test.
Indeed, the experiments show that the equality of co-probabilities is not rejected using uniform splitting.
This statement could also be proven analytically using basic probability notions, but it is not the focus of this appendix.

Previous works \cite{damie_highly_2021,dittert_too_2023} experimentally proved the third step for the refined score attack.
Since the $\epsilon$ measures the quality of the attacker knowledge, we can interpret the statement as follows: more precise attacker knowledge leads to higher accuracy.
We can be convinced this statement is true for all attacks because the opposite sounds nonsensical.

\subsection{Definitions}
\paragraph{Notation}
Let $\mathbb{E}\left[X\right]$ be the expected value of the random variable $X$, and $\var(X)$ be its variance.
The sign $\xrightarrow[]{\mathbb{P}}$ denotes the convergence in probability, and the sign $\xrightarrow[]{(d)}$ the convergence in distribution.
We use the notation $\delta_a$ for the Dirac distribution on point $a \in \mathbb{R}$.

\paragraph{Stochastic dominance definition}
Our mathematical analysis relies on ``stochastic dominance'' \cite{hadar_rules_1969}, a partial order between random variables.
Let the sign $\preccurlyeq$ define this partial order:

\begin{definition}
    A random variable $A$ has first-order stochastic dominance over random variable B if $\forall x, \mathbb{P}(A\le x) \ge \mathbb{P}(B\le x)$.
\end{definition}

Our analysis focuses only on first-order stochastic dominance.
We rely on the alternative definition presented in Definition \ref{def:alt-domi} to prove the stochastic dominance.

\begin{definition}
    \label{def:alt-domi}
    A random variable $A$ has first-order stochastic dominance over random variable B if, for a utility function $u$ continuous, bounded, and increasing, we have $\mathbb{E}[u(A)] \ge \mathbb{E}[u(B)]$.
\end{definition}

\textbf{$\epsilon$-similarity distribution}
Let $\mathcal{E}^{p_\text{ind}, p_\text{atk}}$ define the random
probability distribution of the $\epsilon$-similarity.
Since $\epsilon = \mynorm{\text{SimMat}}$, we have\footnote{This equation manipulates random probability distributions $\mathcal{C}^{n_\text{atk}, p_\text{atk}}$ and $\mathcal{C}^{n_\text{ind}, p_\text{ind}}$ contrary to  Equation \eqref{eq:sim_mat_def} that manipulates classic matrices.}:
$$\mathcal{E}^{p_\text{ind}, p_\text{atk}} = \mynorm{\frac{\mathcal{C}^{n_\text{ind}, p_\text{ind}}}{n_\text{ind}} - \frac{\mathcal{C}^{n_\text{atk}, p_\text{atk}}}{n_\text{atk}}}$$
The parameters $p_\text{ind}$ and $p_\text{atk}$ are the parameters of the probability distributions from which $ D_\text{ind}$ and $ D_\text{atk}$ are drawn.

\subsection{Stochastic advantage}
\label{subsec:stoch_adv}

We want to show that having equal co-probabilities on the attacker and the indexed datasets leads to smaller $\epsilon$.
Mathematically, we want to compare the following random distributions: $\mathcal{E}^{p_\text{ind}, p_\text{ind}}$ (i.e., equality of co-probabilities) and $\mathcal{E}^{p_\text{ind}, p_\text{atk}}$.
Appendix \ref{app:proof} proves that asymptotically (when the dataset sizes tend to infinity):
\begin{equation}
    \label{eq:domi}
    \mathcal{E}^{p_\text{ind}, p_\text{ind}} \preccurlyeq \mathcal{E}^{p_\text{ind}, p_\text{atk}}
\end{equation}

This dominance result implies that it is more likely to reach lower $\epsilon$ values when the attacker's dataset is drawn from a distribution with the same probabilities $p_{ij}$ as the distribution from which the indexed dataset has been drawn.

\subsection{Stochastic dominance proof}
\label{app:proof}
The previous notations were simplified, so we redefine them more precisely here.
First, the dataset sizes are now considered random variables\footnote{This hypothesis is not restrictive: a deterministic sequence is a special case of a sequence of random variables.}.
Let $\p{n^\text{ind}_n}_n$ be a sequence of random variables such that $\lim_{n\rightarrow +\infty} n^\text{ind}_n=+\infty$.
Analoguously, we have a sequence of random variables $\p{n^\text{atk}_n}_n$ be a sequence of random variables such that $\lim_{n\rightarrow +\infty} n^\text{atk}_n= +\infty$.
We can redefine the distribution of the co-occurrence matrices from these variables: $\forall i,j\in \{1\dots m\}$, let $\mathcal{C}^{\text{ind}, p^\text{ind}}_{ij,n} \sim \mathcal{B}\p{n^\text{ind}_n, p^\text{ind}_{ij}}$ (resp. $\mathcal{C}^{\text{atk}, p^\text{atk}}_{ij,n} \sim \mathcal{B}\p{n^\text{atk}_n, p^\text{atk}_{ij}}$) be the random probability distribution of the co-occurrence matrix of the indexed (resp. attacker) dataset.
We note $\mathcal{C}^{\text{ind}, p^\text{ind}}_{\cdot,n}$ (resp. $\mathcal{C}^{\text{atk}, p^\text{atk}}_{\cdot,n}$) the complete matrix distribution (i.e., $\mathcal{C}^{\text{ind}, p^\text{ind}}_{ij,n}$ is the distribution of the $i,j$ variable of $\mathcal{C}^{\text{ind}, p^\text{ind}}_{\cdot,n}$).
We assume that $\mathcal{C}^{\text{ind}, p^\text{ind}}_{ij,n}$ and $\mathcal{C}^{\text{atk}, p^\text{atk}}_{ij,n}$ are independent, but we do not suppose independence for $\mathcal{C}^{\text{ind}, p^\text{ind}}_{ij,n}$ and $\mathcal{C}^{\text{ind}, p^\text{ind}}_{i'j',n}$ (same for $\mathcal{C}^{\text{atk}, p^\text{atk}}_{ij,n}$ and $\mathcal{C}^{\text{atk}, p^\text{atk}}_{i'j',n}$).
The $\epsilon$-similarity probability distribution is then:
$$\mathcal{E}^{p^\text{ind}, p^\text{atk}}_n = \mynorm{\frac{\mathcal{C}^{\text{ind}, p^\text{ind}}_{\cdot,n}}{n^\text{ind}_n}- \frac{\mathcal{C}^{\text{atk}, p^\text{atk}}_{\cdot,n}}{n^\text{atk}_n}}_2$$

Using this notation, we can write Theorem \ref{th:domi} (equivalent to Equation \eqref{eq:domi}).

\begin{theorem}
    \label{th:domi}
    Asymptotically, we have $ \mathcal{E}^{p^\text{ind}, p^\text{ind}}_n \preccurlyeq  \mathcal{E}^{p^\text{ind}, p^\text{atk}}_n$.
\end{theorem}
\begin{proof}
    We have $\mathbb{E}\bracket{\frac{\mathcal{C}^{\text{ind}, p^\text{ind}}_{ij,n}}{n^\text{ind}_n}} = p^\text{ind}_{ij}$ and $Var\p{\frac{\mathcal{C}^{\text{ind}, p^\text{ind}}_{ij,n}}{n^\text{ind}_n}} \rightarrow 0$ when $n \rightarrow \infty$.
    Using Chebyshev's inequality, we deduce that
    \begin{equation}
        \label{eq:conv_frac_ind}
        \frac{\mathcal{C}^{\text{ind}, p^\text{ind}}_{ij,n}}{n^\text{ind}_n} \xrightarrow[n\rightarrow \infty]{\mathbb{P}} p^\text{ind}_{ij}
    \end{equation}
    Analogously, we obtain:
    \begin{equation}
        \label{eq:conv_frac_atk}
        \frac{\mathcal{C}^{\text{atk}, p^\text{atk}}_{ij,n}}{n^\text{atk}_n} \xrightarrow[n\rightarrow \infty]{\mathbb{P}} p^\text{atk}_{ij}
    \end{equation}
    Then, we use the continuous mapping theorem with the continuous function $f(x,y) = \p{x - y}^2$ on Equations \eqref{eq:conv_frac_ind} and \eqref{eq:conv_frac_atk} to obtain:
    \begin{equation}
        \p{\frac{\mathcal{C}^{\text{atk}, p^\text{atk}}_{ij,n}}{n^\text{atk}_n} - \frac{\mathcal{C}^{\text{ind}, p^\text{ind}}_{ij,n}}{n^\text{ind}_n}}^2 \xrightarrow[n\rightarrow \infty]{\mathbb{P}} \p{p^\text{atk}_{ij} - p^\text{ind}_{ij}}^2
    \end{equation}

    We reuse the continuous mapping on this result but with the continuous function $g(z) = \sqrt{\sum_{k=1} z_k}$ to obtain:
    \begin{align}
        \label{eq:epsilon_convergence}
         & \sqrt{\sum_{(i,j) \in \{1\dots m\}^2}  \p{\frac{\mathcal{C}^{\text{atk}, p^\text{atk}}_{ij,n}}{n^\text{atk}_n} - \frac{\mathcal{C}^{\text{ind}, p^\text{ind}}_{ij,n}}{n^\text{ind}_n}}^2} \xrightarrow[n\rightarrow \infty]{\mathbb{P}} \sqrt{\sum_{(i,j) \in \{1\dots m\}^2}\p{p^\text{atk}_{ij} - p^\text{ind}_{ij}}^2} \\
         & \quad \iff      \mathcal{E}_n^{p^\text{atk}, p^\text{ind}} \xrightarrow[n\rightarrow \infty]{\mathbb{P}} \mynorm{p^\text{atk} - p^\text{ind}}                                                                                                                                                                             \\
         & \quad \implies  \mathcal{E}_n^{p^\text{atk}, p^\text{ind}} \xrightarrow[n\rightarrow \infty]{(d)} \delta_{\mynorm{p^\text{atk} - p^\text{ind}}}
    \end{align}

    Let $u$ be a utility function $u$ continuous, bounded, and increasing. From the definition of the convergence in distribution, we deduce that:
    \begin{align}
        \lim_{n \rightarrow \infty} \mathbb{E}\bracket{u\p{\mathcal{E}_n^{p^\text{atk}, p^\text{ind}}}} & = \int ud\delta_{\mynorm{p^\text{atk} - p^\text{ind}}} \\
                                                                                                        & = u(\mynorm{p^\text{atk} - p^\text{ind}})
    \end{align}
    We remark that $\lim_{n \rightarrow \infty} \mathbb{E}\bracket{u\p{\mathcal{E}_n^{p^\text{ind}, p^\text{ind}}}}=u(0)$ and that $u$ is increasing so:
    \begin{align}
             & u(0) \le u(\mynorm{p^\text{atk} - p^\text{ind}})                                                                                                                                                    \\
        \iff & \lim_{n \rightarrow \infty} \mathbb{E}\bracket{u\p{\mathcal{E}_n^{p^\text{ind}, p^\text{ind}}}} \le \lim_{n \rightarrow \infty} \mathbb{E}\bracket{u\p{\mathcal{E}_n^{p^\text{atk}, p^\text{ind}}}}
        \label{eq:ineq_lim_expected_utility}
    \end{align}
    From Equation \eqref{eq:ineq_lim_expected_utility}, we use the Definition \ref{def:alt-domi} of stochastic dominance to conclude that asymptotically: $\mathcal{E}^{p^\text{ind}, p^\text{ind}}_n \preccurlyeq  \mathcal{E}^{p^\text{ind}, p^\text{atk}}_n$
\end{proof}

\section{Proving the distribution shift of the Apache dataset}
\label{app:splitting-equality-coprob}

Section \ref{sec:dataset_generation} highlighted that uniform splitting could lead to more powerful attackers than temporal splitting.
We want to deepen this result and highlight the difference between these two setups: the equality of co-probabilities.
This appendix builds a statistical test to show that, contrary to the uniform split, the temporal split (on the Apache dataset) does not generate a dataset distribution verifying the equality of co-probabilities.
In other words, we rigorously prove the distribution shift over time in the Apache dataset.

\paragraph{Statistics background}
Statistical tests verify a hypothesis based on observed data.
When a statistical test is ``rejected'', the hypothesis is then false with high probability.
These statistical tools are standard in medical research to prove various statements (e.g., vaccine efficacy).

Our goal is then to conceive a test for our equality assumption.
When performing a statistical test, the analysis is focused on the $p$-value (with $p\in [0,1]$).
A $p$-value corresponds to the probability of obtaining our observed data under a certain \textit{null hypothesis}, referred to as $H_0$: the equality of co-probabilities in our case.
The opposite hypothesis is referred to as $H_1$: the inequality of at least one of the co-probabilities in our case.
The lower $p$ is, the more confident we are that the hypothesis $H_0$ is false.
If $p$ is considered negligible, the null hypothesis is rejected.
The ``negligibility'' threshold can be 0.01, 0.001, or even lower, depending on the research field.

The Z-test for the equality of two proportions \cite{kanji_100_2006} is the standard statistical test to check the equality of two probabilities.
In our case, we need to test the equality of $\frac{m(m+1)}{2}$ pair of probabilities (i.e., $p^\text{ind}_{ij}$ and $p^\text{atk}_{ij}$).
To build our test, we propose first to test all the pairs of co-probabilities individually and then combine the result of these tests into a unique $p$-value.

\paragraph{Individual Z-tests}
We use the Z-test to individually test the equality of each co-probability $p_{ij}$.
Each Z-test has as hypothesis $H_0: p_{ij}^\text{ind}=p_{ij}^\text{atk}$ (and $H_1: p_{ij}^\text{ind}\neq p_{ij}^\text{atk}$).
For details about the Z-test, we refer to Kanji \cite{kanji_100_2006}.
The Z-test for proportions is widely implemented in many languages, including R\footnote{https://www.rdocumentation.org/packages/stats/versions/3.6.2/topics/prop.test}.

For clarity, we refer to the $p$-values as $pv$ to avoid confusion with the co-probabilities already using the notation $p_{ij}$.
Then, for each pair of keyword $(w_i, w_j)$, we obtain a $p$-value $pv_{ij}$ (for all $i,j \in [m]$).
Note that $pv_{ij} = pv_{ji}$ so we have $\frac{m(m+1)}{2}$ unique $p$-values.

\paragraph{Combining multiple $p$-values}
We now have $\frac{m(m+1)}{2}$ \text{dependent} $p$-values, and we need to combine all of them to create a test for our complete hypothesis stated in the equality of co-probabilities.
In other words, we want a test verifying whether all the sub-hypotheses are true simultaneously, sub-hypotheses for which we have individual tests.
This problem is known as the multiple comparisons problem \cite{miller_simultaneous_2012}.
It is non-trivial because these $p$-values are dependent.
The simplest solution to this problem is the Bonferroni correction \cite{bland_multiple_1995,weisstein_bonferroni_2004}.

The Bonferroni correction takes as an input $M$ (possibly dependent) $p$-values and outputs a $p$-value for the combination of all hypotheses.
The ``corrected'' $p$-value corresponds to the minimum $p$-value in the input set multiplied by $M$.
We call this a \emph{corrected} $p$-value because it is not formally a $p$-value (e.g., it can be above 1).
However, it is common to interpret the output of a Bonferroni correction as a $p$-value.
In our case, we can write the corrected $p$-value as follows:
\begin{equation}
    \widetilde{pv} = \frac{m(m+1)}{2} \times \min_{(i,j) \in \{1\dots m\}^2} pv_{ij}
\end{equation}

The corrected $p$-value is proportional to the minimum of the $M$ initial $p$-values, so this test is highly sensitive.
The risk of sensitive test metrics is constantly rejecting (or not rejecting) the null hypothesis for any dataset.
In such a case, we could not draw any conclusions.
Other solutions to the multiple comparisons problem are less sensitive than the Bonferroni correction \cite{holm1979simple,hsu1996multiple,miller_simultaneous_2012}.
However, we limited our analysis to this simple Bonferroni correction, sufficient to observe the phenomenon we want to highlight: an equality and a non-equality case.

\paragraph{Experimental results} To verify our distribution shift claim, we compute $\widetilde{pv}$ using the temporal split (Table \ref{tab:sim-per-year}) and the uniform split (Table \ref{tab:sim-rand}).

\begin{table}[t]
    \caption{Statistical tests of the equality of co-probabilities on the Apache dataset using two splitting methods.}
    \begin{subtable}{0.45\textwidth}
        \centering
        \begin{tabular}{|c|c|c|c|c|}\hline
            Year             & 2003 & 2005 & 2007 & 2009 \\\hline
            $\widetilde{pv}$ & 0.0  & 0.0  & 0.0  & 0.0  \\\hline
        \end{tabular}
        \caption{Temporal split}
        \label{tab:sim-per-year}
    \end{subtable}
    \begin{subtable}{0.45\textwidth}
        \centering
        \begin{tabular}{|c|c|c|c|}\hline
                             & Avg. & Min. & Max.  \\\hline
            $\widetilde{pv}$ & 2.40 & 0.02 & 11.03 \\ \hline
        \end{tabular}
        \caption{Uniformly random split (100 repetitions).}
        \label{tab:sim-rand}
    \end{subtable}
\end{table}

With the temporal split (Table \ref{tab:sim-per-year}), the corrected $p$-values are extremely small, so the equality of co-probabilities is strongly rejected.
Theoretically, the $p$-values should never be equal to zero.
The values of the first experiment are so low (i.e., below machine epsilon) that they were automatically rounded to zero.
This result proves the distribution shift in the Apache dataset.

Using the uniform split (Table \ref{tab:sim-rand}), the equality of co-probabilities is not rejected because the $p$-values are generally high.
While we can obtain this equality result from basic mathematical analysis, this non-rejection also proves that the rejection for the temporal split is due to the splitting method properties and not the test sensitivity.
We repeat the uniform split 100 times to show the temporal split was not somehow ``unlucky'': even the uniform split with the worst $\widetilde{pv}$ (denoted as Min. $\widetilde{pv}$ in Table \ref{tab:sim-rand}) has a $p$-value way above all the results obtained using the temporal split.

\end{document}